\documentclass{llncs}

\usepackage[usenames,dvipsnames]{color}

\usepackage{latexsym}
\usepackage{amsxtra} 
\usepackage{amssymb}
\usepackage{amsmath}
\usepackage{pslatex}
\usepackage{epsfig}
\usepackage{wrapfig}
\usepackage{paralist}
\usepackage{graphics}
\usepackage{stmaryrd}
\usepackage{txfonts}
\usepackage{framed}
\usepackage{makecell}
\usepackage{url}

\usepackage{proof}

\usepackage{algorithm}
\usepackage{algorithmicx}
\usepackage[noend]{algpseudocode}

\newtheorem{fact}{Fact}

\newcommand{\set}[1]{\left\{ #1 \right\}}
\newcommand{\tuple}[1]{\left\langle #1 \right\rangle}
\renewcommand{\vec}[1]{\mathbf #1}

\newcommand{\len}[1]{{|{#1}|}}
\newcommand{\card}[1]{{|\!|{#1}|\!|}}

\newcommand{\arrow}[2]{\xrightarrow{{\scriptstyle #1}}_{{\scriptstyle #2}}}

\newcommand{\nat}{{\bf \mathbb{N}}}
\newcommand{\zed}{{\bf \mathbb{Z}}}

\renewcommand{\paragraph}[1]{\noindent{\bf #1}}

\newif\ifLongVersion\LongVersiontrue
\newcommand{\ssorts}[1]{#1^\mathrm{s}}
\newcommand{\sfuns}[1]{#1^\mathrm{f}}
\newcommand{\mods}{\mathbf{I}}

\newcommand{\teq}{\approx}

\newcommand{\I}{\mathcal{I}}

\newcommand{\locs}{\mathsf{Loc}}
\newcommand{\data}{\mathsf{Data}}
\newcommand{\nil}{\mathsf{nil}}

\newcommand{\emp}{\mathsf{emp}}
\newcommand{\wand}{
 \mathrel{\mbox{$\hspace*{-0.03em}\mathord{-}\hspace*{-0.66em}
 \mathord{-}\hspace*{-0.36em}\mathord{*}$\hspace*{-0.005em}}}} % {\multimap}
\newcommand{\seplog}{\mathsf{SL}}
\newcommand{\tinyseplog}{\mathsf{\scriptscriptstyle{SL}}}

\newcommand{\tterm}{\mathsf{t}}
\newcommand{\uterm}{\mathsf{u}}
\newcommand{\vterm}{\mathsf{v}}
\newcommand{\wterm}{\mathsf{w}}
\newcommand{\xterm}{\mathsf{x}}
\newcommand{\yterm}{\mathsf{y}}
\newcommand{\zterm}{\mathsf{z}}
\newcommand{\nextp}{\mathsf{next}}
\newcommand{\datap}{\mathsf{data}}

\newcommand{\fv}[1]{\mathrm{Fvc}(#1)}

\newcommand{\dom}{\mathrm{dom}}

\newcommand{\vars}{\mathsf{Vars}}

\newcommand{\funcsolve}{\small\mathsf{solve}}
\newcommand{\funcsmtsolve}{\small\mathsf{solve\_rec}}

\newcommand{\Int}{\mathsf{Int}}

\newcommand{\euf}{\mathsf{E}}
\newcommand{\lia}{\mathsf{LIA}}
\newcommand{\euflia}{\mathsf{ELIA}}

\begin{document}
%%%%%%%%%%%%%%%%%%%%%%%%%%%%%%%%%%%%%%%%%%%%%%%%%%%%%%%%%%%%%%%%%%%%%%%%%%%%%%%

\title{Reasoning in the Bernays-Sch\"onfinkel-Ramsey Fragment of
  Separation Logic}

\author{Andrew Reynolds\inst{1} \and Radu Iosif\inst{2} \and Cristina Serban\inst{2}} 

\institute{The University of Iowa \and Verimag/CNRS/Universit\'e de Grenoble Alpes}

\maketitle

\begin{abstract}
  Separation Logic ($\seplog$) is a well-known assertion language used
  in Hoare-style modular proof systems for programs with dynamically
  allocated data structures. In this paper we investigate the fragment
  of first-order $\seplog$ restricted to the
  Bernays-Sch\"onfinkel-Ramsey quantifier prefix $\exists^*\forall^*$,
  where the quantified variables range over the set of memory
  locations. When this set is uninterpreted (has no associated theory)
  the fragment is \textsc{PSPACE}-complete, which matches the
  complexity of the quantifier-free fragment
  \cite{CalcagnoYangOHearn01}. However, $\seplog$ becomes undecidable
  when the quantifier prefix belongs to $\exists^*\forall^*\exists^*$
  instead, or when the memory locations are interpreted as integers
  with linear arithmetic constraints, thus setting a sharp boundary
  for decidability within $\seplog$. We have implemented a decision
  procedure for the decidable fragment of $\exists^*\forall^*\seplog$
  as a specialized solver inside a DPLL($T$) architecture, within the
  CVC4 SMT solver. The evaluation of our implementation was carried
  out using two sets of verification conditions, produced
  by\begin{inparaenum}[(i)] \item unfolding inductive predicates,
  and \item a weakest precondition-based verification condition
  generator. \end{inparaenum} Experimental data shows that automated
  quantifier instantiation has little overhead, compared to manual
  model-based instantiation.
\end{abstract}

\section{Introduction}\label{sec:intro}

Separation Logic ($\seplog$) is a popular logical framework for
program verification, used by a large number of methods, ranging from
static analysis \cite{Predator,Xisa,Infer} to Hoare-style proofs
\cite{Sleek} and property-guided abstraction refinement
\cite{SplInter}. The salient features that make $\seplog$ particularly
attractive for program verification are the ability of
defining\begin{inparaenum}[(i)] \item recursive data structures using
  small and natural inductive definitions, \item weakest pre- and
  post-condition calculi that capture the semantics of programs with
  pointers, and \item compositional verification methods, based on the
  principle of local reasoning (analyzing separately pieces of program
  working on disjoint heaps). \end{inparaenum}

Consider, for instance, the following inductive definitions,
describing an acyclic and a possibly cyclic list segment,
respectively:
\[\begin{array}{ll}
\widehat{\mathsf{ls}}(\xterm,\yterm) \equiv \emp \wedge \xterm=\yterm ~\vee~ 
\xterm \neq \yterm \wedge \exists \zterm ~.~ \xterm \mapsto \zterm * \widehat{\mathsf{ls}}(\zterm,\yterm) & 
\text{ acyclic list segment from $\xterm$ to $\yterm$}
\\
\mathsf{ls}(\xterm,\yterm) \equiv \emp \wedge \xterm=\yterm ~\vee~ 
\exists \uterm ~.~ \xterm \mapsto \uterm * \mathsf{ls}(\uterm,\yterm) 
& \text{ list segment from $\xterm$ to $\yterm$}
\end{array}\]
Intuitively, an acyclic list segment is either empty, in which case
the head and the tail coincide ($\emp \wedge \xterm=\yterm$), or it
contains at least one element which is disjoint from the rest of the
list segment. We denote by $\xterm \mapsto \zterm$ the fact that
$\xterm$ is an allocated memory location, which points to $\zterm$,
and by $\xterm \mapsto \zterm * \widehat{\mathsf{ls}}(\zterm,\yterm)$
the fact that $\xterm \mapsto \zterm$ and
$\widehat{\mathsf{ls}}(\zterm,\yterm)$ hold over disjoint parts of the
heap. The constraint $\xterm\neq\yterm$, in the inductive definition
of $\widehat{\mathsf{ls}}$, captures the fact that the tail of the
list segment is distinct from every allocated cell in the list
segment, which ensures the acyclicity condition. Since this constraint
is omitted from the definition of the second (possibly cyclic) list
segment $\mathsf{ls}(\xterm,\yterm)$, its tail $\yterm$ is allowed to
point inside the set of allocated cells.

Automated reasoning is the key enabler of push-button program
verification. Any procedure that checks the validity of a logical
entailment between inductive predicates requires checking the
satisfiability of formulae from the base (non-inductive) assertion
language, as shown by the example below. Consider a fragment of the
inductive proof showing that any acyclic list segment is also a list
segment, given below:
\[\infer[]{\widehat{\mathsf{ls}}(\xterm,\yterm) \vdash \mathsf{ls}(\xterm,\yterm)}{
  \infer[\begin{array}{l}\xterm \neq \yterm \wedge \xterm \mapsto
      \zterm \models \exists \uterm ~.~ \xterm \mapsto \uterm
      \\ \text{by instantiation } \uterm \leftarrow
      \zterm \end{array}]{ \xterm \neq \yterm \wedge \xterm \mapsto
    \zterm * \widehat{\mathsf{ls}}(\zterm,\yterm) \vdash \exists
    \uterm ~.~ \xterm \mapsto \uterm * \mathsf{ls}(\uterm,\yterm)
  }{\widehat{\mathsf{ls}}(\zterm,\yterm) \vdash
    \mathsf{ls}(\zterm,\yterm)}}\] The first (bottom) inference in
the proof corresponds to one of the two cases produced by unfolding
both the antecedent and consequent of the entailment (the second case
$\emp \wedge \xterm=\yterm \vdash \emp \wedge \xterm=\yterm$ is
trivial and omitted for clarity). The second inference is a
simplification of the sequent obtained by unfolding, to a
sequent matching the initial one (by renaming $\zterm$ to $\xterm$),
and allows to conclude this branch of the proof by an inductive
argument, based on the principle of infinite descent
\cite{BrotherstonSimpson11}.

The simplification applied by the second inference above relies on
the validity of the entailment $\xterm \neq \yterm \wedge \xterm
\mapsto \zterm \models \exists \uterm ~.~ \xterm \mapsto \uterm$,
which reduces to the (un)satisfiability of the formula $\xterm \neq
\yterm \wedge \xterm \mapsto \zterm \wedge \forall \uterm ~.~ \neg
\xterm \mapsto \uterm$. The latter falls into the
Bernays-Sch\"onfinkel-Ramsey fragment, defined by the
$\exists^*\forall^*$ quantifier prefix, and can be proved
unsatisfiable using the instantiation of the universally quantified
variable $\uterm$ with the existentially quantified variable $\zterm$
(or a corresponding Skolem constant). In other words, this formula is
unsatisfiable because the universal quantified subformula asks that no
memory location is pointed to by $\xterm$, which is contradicted by
$\xterm \mapsto \zterm$. The instantiation of $\uterm$ that violates
the universal condition is $\uterm\leftarrow\zterm$, which is carried
over in the rest of the proof.

The goal of this paper is mechanizing satisfiability of the
Bernays-Sch\"onfinkel-Ramsey fragment of $\seplog$, without
inductively defined predicates\footnote{Strictly speaking, the
  Bernays-Sch\"onfinkel-Ramsey class refers to the
  $\exists^*\forall^*$ fragment of first-order logic with equality and
  predicate symbols, but no function symbols \cite{Lewis80}.}. This
fragment is defined by the quantifier prefix of the formulae in prenex
normal form. We consider formulae $\exists x_1 \ldots \exists x_m
\forall y_1 \ldots \forall y_n ~.~
\phi(x_1,\ldots,x_m,y_1,\ldots,y_n)$, where $\phi$ is any
quantifier-free formula of $\seplog$, consisting of pure formulae from
given base theory $T$, and points-to atomic propositions relating
terms of $T$, combined with unrestricted Boolean and separation
connectives, and the quantified variables range essentially over the
set of memory locations. In a nutshell, the contributions of the paper
are two-fold:
\begin{compactenum}
\item We draw a sharp boundary between decidability and
  undecidability, proving essentially that the satisfiability problem
  for the Bernays-Sch\"onfinkel-Ramsey fragment of $\seplog$ is
  \textsc{PSPACE}-complete, if the domain of memory locations is an
  uninterpreted set, whereas interpreting memory locations as integers
  with linear arithmetic constraints, leads to
  undecidability. Moreover, undecidability occurs even for
  uninterpreted memory locations, if we extend the quantifier prefix
  to $\exists^*\forall^*\exists^*$.
\item We have implemented an effective decision procedure for
  quantifier instantiation, based on counterexample-driven learning of
  conflict lemmas, integrated within the DPLL($T$) architecture
  \cite{GanzingerHagenNieuwenhuisOliverasTinelli04} of the CVC4 SMT
  solver \cite{CVC4-CAV-11}. Experimental evaluation of our
  implementation shows that the overhead of the push-button quantifier
  instantiation is negligible, compared to the time required to solve
  a quantifier-free instance of the problem, obtained manually, by
  model inspection.
\end{compactenum} 

\vspace*{\baselineskip}
\paragraph{\bf Related Work}
The first theoretical results on the decidability and computational
complexity of $\seplog$ (without inductive definitions) were found by
Calcagno, Yang and O'Hearn \cite{CalcagnoYangOHearn01}. They showed that
the satisfiability problem for $\seplog$ is undecidable, in the
presence of quantifiers, assuming that each memory location can point
to two other locations, i.e.\ using atomic propositions of the form
$\xterm \mapsto (\yterm,\zterm)$. Decidability can be recovered by
considering the quantifier-free fragment, proved to be
\textsc{PSPACE}-complete, by a small model argument
\cite{CalcagnoYangOHearn01}. Refinements of these results consider
decidable fragments of $\seplog$ with one record field (atomic
points-to propositions $\xterm \mapsto \yterm$), and one or two
quantified variables. In a nutshell, $\seplog$ with one record field
and separating conjunction only is decidable with non-elementary time
complexity, whereas adding the magic wand adjoint leads to
undecidability \cite{BrocheninDemriLozes11}. Decidability, in the
presence of the magic wand operator, is recovered by restricting the
number of quantifiers to one, in which case the logic becomes
\textsc{PSPACE}-complete \cite{DemriGalmicheWendlingMery14}. This
bound is sharp, because allowing two quantified variables leads to
undecidability, and decidability with non-elementary time complexity
if the magic wand is removed \cite{DemriDeters15}.

SMT techniques were applied to deciding the satisfiability of
$\seplog$ in the work of Piskac, Wies and Zufferey
\cite{Piskac2013,Piskac2014}. They considered quantifier-free
fragments of $\seplog$ with separating conjunction in positive form
(not occurring under negation) and without magic wand, and allowed for
hardcoded inductive predicates (list and tree segments). In a similar
spirit, we previously defined a translation to multi-sorted
second-order logic combined with counterexample-driven instantiation
for set quantifiers to define a decision procedure for the
quantifier-free fragment of $\seplog$ \cite{ReynoldsIosifKingSerban16}. In
a different vein, a tableau-based semi-decision procedure is given by
M\'ery and Galmiche \cite{MeryGalmiche07}. Termination of this
procedure is guaranteed for the (decidable) quantifier-free fragment
of $\seplog$, yet no implementation is available for comparison.

A number of automated theorem provers have efficient and complete approaches for 
the Bernays-Sch\"onfinkel-Ramsey fragment of first-order-logic,
also known as effectively propositional logic (EPR)~\cite{DBLP:journals/ijait/BaumgartnerFT06,DBLP:conf/cade/Korovin08}.
A dedicated approach for EPR in the SMT solver Z3 was developed in~\cite{DBLP:journals/jar/PiskacMB10}.
An approach based on finite model finding is implemented in CVC4~\cite{ReyEtAl-1-RR-13},
which is model-complete for EPR.
Our approach is based on counterexample-guided quantifier instantiation,
which has been used in the context of SMT solving in previous works~\cite{GeDeM-CAV-09,ReynoldsDKBT15Cav}.
%Bernays-Sch\"onfinkel-Ramsey class is decidable in presence of
%integers with difference bounds arithmetic \cite{WeidenbachVoigt15}

\section{Preliminaries}

We consider formulae in multi-sorted first-order logic.
%, over a
%signature consisting of a countable set of sort symbols and a set of
%function symbols, and we write $\top$ and $\bot$ for the Boolean
%constants \emph{true} and \emph{false}. 
A \emph{signature} $\Sigma$ consists of a set $\ssorts{\Sigma}$ of
sort symbols and a set $\sfuns{\Sigma}$ of (sorted) \emph{function
  symbols} $f^{\sigma_1 \cdots \sigma_n \sigma}$, where $n \geq 0$ and
$\sigma_1, \ldots, \sigma_n, \sigma \in \ssorts{\Sigma}$. If $n=0$, we
call $f^\sigma$ a \emph{constant symbol}.  In this paper, we consider
signatures $\Sigma$ containing the Boolean sort, and write $\top$ and
$\bot$ for the Boolean constants \emph{true} and \emph{false}. For
this reason, we do not consider predicate symbols as part of a
signature, as predicates are viewed as Boolean functions.
Additionally, we assume for any finite sequence of sorts $\sigma_1,
\ldots, \sigma_n \in \ssorts{\Sigma}$, the \emph{tuple} sort $\sigma_1
\times \ldots \times \sigma_n$ also belongs to $\ssorts{\Sigma}$,
and that $\sfuns{\Sigma}$ includes the 
$i^{th}$ tuple projection function for each $i = 1, \ldots, n$.
For each $k > 0$, let $\sigma^k$ denote the $k$-tuple sort $\sigma \times
\ldots \times \sigma$.

Let $\vars$ be a countable set of first-order variables, each $x^\sigma \in
\vars$ having an associated sort $\sigma$. First-order terms and formulae
over the signature $\Sigma$ (called $\Sigma$-terms and
$\Sigma$-formulae) are defined as usual. For a $\Sigma$-formula
$\varphi$, we denote by $\fv{\varphi}$ the set of free variables and
constant symbols in $\varphi$, and by writing $\varphi(x)$ we mean
that $x \in \fv{\phi}$. Whenever $\fv{\phi} \cap \vars = \emptyset$,
we say that $\phi$ is a \emph{sentence}, i.e.\ $\phi$ has no free
variables.
A \emph{$\Sigma$-interpretation $\I$} maps:\begin{inparaenum}[(1)]
\item each sort symbol $\sigma \in \Sigma$ to a non-empty set $\sigma^\I$,
\item each function symbol $f^{\sigma_1,\ldots,\sigma_n,\sigma} \in \Sigma$ to a
  total function $f^\I : \sigma^\I_1 \times \ldots \times \sigma^\I_n
  \rightarrow \sigma^\I$ where $n > 0$, and to an element of $\sigma^\I$ when $n
  = 0$, and
\item each variable $x^\sigma \in \vars$ to an element of $\sigma^\I$.
\end{inparaenum}
For an interpretation $\I$ a sort symbol $\sigma$ and a variable $x$,
we denote by $\I[\sigma \leftarrow S]$ and, respectively $\I[x
  \leftarrow v]$, the interpretation associating the set $S$ to
$\sigma$, respectively the value $v$ to $x$, and which behaves like
$\I$ in all other cases\footnote{By writing $\I[\sigma \leftarrow S]$
  we ensure that all variables of sort $\sigma$ are mapped by $\I$ to
  elements of $S$.}. For a $\Sigma$-term $t$, we write $t^\I$ to
denote the interpretation of $t$ in $\I$, defined inductively, as
usual. A satisfiability relation between $\Sigma$-interpretations and
$\Sigma$-formulas, written $\I \models \varphi$, is also defined
inductively, as usual. We say that $\I$ is \emph{a model of $\varphi$}
if $\I$ satisfies $\varphi$.

A (multi-sorted first-order) \emph{theory} is a pair \(T = (\Sigma,
\mods)\) where $\Sigma$ is a signature and $\mods$ is a non-empty set
of $\Sigma$-interpretations, the \emph{models} of $T$. We assume that
$\Sigma$ always contains the equality predicate, which we denote by
$\teq$, as well as projection functions for each tuple sort.  A
$\Sigma$-formula $\varphi$ is \emph{$T$-satisfiable} if it is
satisfied by some interpretation in $\mods$. We write $\euf$ to denote
the empty theory (with equality), whose signature consists of a sort
$U$ with no additional function symbols, and $\lia$ to denote the
theory of linear integer arithmetic, whose signature consists of the
sort $\Int$, the binary predicate symbol $\geq$, function $+$ denoting
addition, and the constants $0,1$ of sort $\Int$, interpreted as
usual. In particular, there are no uninterpreted function symbols in
$\lia$. By $\euflia$ we denote the theory obtained by extending the
signature of $\lia$ with the sort $U$ of $\euf$ and the equality over
$U$.

Let \(T = (\Sigma, \mods)\) be a theory and let $\locs$ and $\data$ be
two sorts from $\Sigma$, with no restriction other than the fact that
$\locs$ is always interpreted as a countable set. Also, we consider
that $\Sigma$ has a designated constant symbol $\nil^\locs$.  The
\emph{Separation Logic} fragment $\seplog(T)_{\locs,\data}$ is the set of
formulae generated by the following syntax:
\[\begin{array}{lcl}
\varphi & := & \phi \mid \emp \mid \tterm \mapsto \uterm \mid
\varphi_1 * \varphi_2 \mid \varphi_1 \wand \varphi_2 \mid \neg
\varphi_1 \mid \varphi_1 \wedge \varphi_2 \mid \exists x^\sigma ~.~
\varphi_1(x)
\end{array}\]
where $\phi$ is a $\Sigma$-formula, and $\tterm$, $\uterm$ are
$\Sigma$-terms of sorts $\locs$ and $\data$, respectively. As usual,
we write $\forall x^\sigma ~.~ \varphi(x)$ for $\neg\exists x^\sigma ~.~
\neg\varphi(x)$. We omit specifying the sorts of variables and
constants when they are clear from the context.

Given an interpretation $\I$, a \emph{heap} is a finite partial
mapping $h : \locs^\I \rightharpoonup_{\mathrm{fin}} \data^\I$. For a
heap $h$, we denote by $\dom(h)$ its domain. For two heaps $h_1$ and
$h_2$, we write $h_1 \# h_2$ for $\dom(h_1) \cap \dom(h_2) =
\emptyset$ and $h = h_1 \uplus h_2$ for $h_1 \# h_2$ and $h = h_1 \cup
h_2$. We define the \emph{satisfaction relation} $\I,h
\models_{\tinyseplog} \phi$ inductively, as follows:
\[\begin{array}{lcl}
\I,h \models_{\tinyseplog} \phi & \iff & \I \models \phi \text{ if $\phi$ is a $\Sigma$-formula} \\
\I,h \models_{\tinyseplog} \emp & \iff & h = \emptyset \\
\I,h \models_{\tinyseplog} \tterm \mapsto \uterm & \iff & 
h = \{(\tterm^\I,\uterm^\I)\} \text{ and } \tterm^\I\not\teq\nil^\I \\
\I,h \models_{\tinyseplog} \phi_1 * \phi_2 & \iff & \text{there exist heaps } h_1,h_2 
\text{ s.t. } h=h_1\uplus h_2 
\text{ and } \I,h_i \models_{\tinyseplog} \phi_i, i = 1,2 \\
\I,h \models_{\tinyseplog} \phi_1 \wand \phi_2 & \iff & \text{for all heaps } h' \text{ if } h'\#h 
\text{ and } \I,h' \models_{\tinyseplog} \phi_1
\text{ then } \I,h'\uplus h \models_{\tinyseplog} \phi_2 \\
\I,h \models_{\tinyseplog} \exists x^S . \varphi(x) & \iff & 
\I[x \leftarrow s],h \models_{\tinyseplog} \varphi(x) \text{, for some }s \in S^\I
\end{array}\]
The satisfaction relation for $\Sigma$-formulae, Boolean connectives
$\wedge$, $\neg$, and linear arithmetic atoms, are the classical ones
from first-order logic. Notice that the range of a quantified variable
$x^S$ is the interpretation of its associated sort $S^\I$.

A formula $\varphi$ is said to be \emph{satisfiable} if there exists
an interpretation $\I$ and a heap $h$ such that $\I,h
\models_{\tinyseplog} \varphi$. The $(\seplog,T)$-\emph{satisfiability
  problem} asks, given an $\seplog$ formula $\varphi$, whether there
exists an interpretation $\I$ of $T$ and a heap $h$ such that $\I,h
\models_{\tinyseplog} \varphi$. We write $\varphi
\models_{\tinyseplog} \psi$ if for every interpretation $\I$ and heap
$h$, if $\I,h \models_{\tinyseplog} \varphi$ then $\I,h
\models_{\tinyseplog} \psi$, and we say that $\varphi$ \emph{entails}
$\psi$ in this case. 

\vspace*{\baselineskip}
\paragraph{\bf The Bernays-Sch\"onfinkel-Ramsey Fragment of $\seplog$}
In this paper we address the satisfiability problem for the class of
sentences $\phi \equiv \exists x_1 \ldots \exists x_m \forall y_1
\ldots \forall y_n ~.~ \varphi(x_1, \ldots, x_m, y_1, \\ \ldots,
y_n)$, where $\varphi$ is a quantifier-free formula of
$\seplog(T)_{\locs,\data}$. We shall denote this fragment by
$\exists^*\forall^*\seplog(T)_{\locs,\data}$. It is easy to see that
any sentence $\phi$, as above, is satisfiable if and only if the
sentence $\forall y_1 \ldots \forall y_n ~.~ \varphi[c_1/x_1, \ldots,
  c_m/x_m]$ is satisfiable, where $c_1,\ldots,c_m$ are fresh (Skolem)
constant symbols. The latter is called the \emph{functional form} of
$\phi$.

As previously mentioned, $\seplog$ is used mainly specify properties
of a program's heap. If the program under consideration uses pointer
arithmetic, as in C or C++, it is useful to consider $\lia$ for the
theory of memory addresses. Otherwise, if the program only compares
the values of the pointers for equality, as in Java, one can use
$\euf$ for this purpose. This distinction led us to considering the
satisfiability problem for
$\exists^*\forall^*\seplog(T)_{\locs,\data}$ in the following cases:
\begin{compactenum}
\item $\locs$ is interpreted as the sort $U$ of $\euf$ and $\data$ as
  $U^k$, for some $k\geq1$. The satisfiability problem for the
  fragment $\exists^*\forall^*\seplog(\euf)_{U,U^k}$ is
  \textsc{PSPACE}-complete, and the proof follows a small model
  property argument.
\item as above, with the further constraint that $U$ is interpreted as
  an infinite countable set, i.e.\ of cardinality $\aleph_0$. In this
  case, we prove a cut-off property stating that all locations not in
  the domain of the heap and not used in the interpretation of
  constants, are equivalent from the point of view of an $\seplog$
  formula. This satisfiability problem is reduced to the unconstrained one
  above, and also found to be \textsc{PSPACE}-complete. 
\item both $\locs$ and $\data$ are interpreted as $\Int$, equipped
  with addition and total order, in which case
  $\exists^*\forall^*\seplog(\lia)_{\Int,\Int}$ is undecidable.
\item $\locs$ is interpreted as the sort $U$ of $\euf$, and $\data$ as $U
  \times \Int$. Then $\exists^*\forall^*\seplog(\euflia)_{U,U \times
    \Int}$ is undecidable.
\end{compactenum}
Additionally, we prove that the fragment
$\exists^*\forall^*\exists^*\seplog(\euf)_{U,U^k}$, with two
quantifier alternations, is undecidable, if $k\geq2$. The question
whether the fragment $\exists^*\forall^*\seplog(\euflia)_{U,\Int}$ is
decidable is currently open, and considered for future work.  

% For space reasons, all missing proofs are given in \cite{EPRArxiv}.

\section{Decidability and Complexity Results}

This section defines the decidable cases of the
Bernays-Sch\"onfinkel-Ramsey fragment of $\seplog$, with matching
undecidable extensions. The decidable fragment
$\exists^*\forall^*\seplog(\euf)_{U,U^k}$ relies on a small model
property given in Section \ref{sec:small-model}. Undecidability of
$\exists^*\forall^*\seplog(\lia)_{\Int,\Int}$ is obtained by a
refinement of the undecidability proof for Presburger arithmetic with
one monadic predicate \cite{Halpern91}, in Section
\ref{sec:undecidability}.

\subsection{Small Model Property}\label{sec:small-model}

The decidability proof for the quantifier-free fragment of $\seplog$
\cite{CalcagnoYangOHearn01,YangPhd} relies on a small model
property. Intuitively, no quantifier-free $\seplog$ formula can
distinguish between heaps in which the number of invisible locations,
not in the range of the set of free variables, exceeds a certain
threshold, linear in the size of the formula. Then a formula is
satisfiable iff it has a heap model of size linear in the size of the
input formula.

For reasons of self-containment, we recall a number of definitions and
results from \cite{YangPhd}. Some of them are slightly modified for
our purposes, but these changes have no effect on the validity of the
original proofs for the Lemmas \ref{lemma:small-model-existence} and
\ref{lemma:heap-equivalence} below. In the rest of this section, we
consider formulae of $\seplog(\euf)_{U,U^k}$, meaning
that\begin{inparaenum}[(i)]
\item $\locs = U$, and \item there exists an integer $k>0$ such that
  $\data = U^k$, where $U$ is the (uninterpreted) sort of
  $\euf$. \end{inparaenum} We fix $k$ for the rest of this section. 

\begin{definition}\cite[Definition 90]{YangPhd}
Given a set of locations $S$, the equivalence relation $=_S$ between
$k$-tuples of locations is defined as $\langle v_1, \ldots, v_k
\rangle =_S \langle v'_1, \ldots, v'_k\rangle$ if and only
if \begin{compactitem}
\item if $v_i \in S$ then $v_i=v'_i$, and
\item if $v_i \not\in S$ then $v'_i \not\in S$, 
\end{compactitem}
for all $i = 1, \ldots, k$. 
\end{definition}
Intuitively, $=_S$ restricts the equality to the elements in
$S$. Observe that $=_S$ is an equivalence relation and that $S
\subseteq T$ implies $=_T ~\subseteq~ =_S$. For a set $S$, we write
$\card{S}$ for its cardinality, in the following.

\begin{definition}\cite[Definition 91]{YangPhd}\label{def:heap-equivalence}
Given an interpretation $\I$, an integer $n>0$, a set of variables $X
\subseteq \vars$ and a set of locations $S \subseteq U^\I$, for any two
heaps $h,h' : U^\I \rightharpoonup_{\mathrm{fin}} (U^\I)^k$, we define $h
\sim^\I_{n,X,S} h'$ if and only if \begin{compactenum}
\item $\I(X) \cap \dom(h) = \I(X) \cap \dom(h')$, 
\item for all $\ell \in \I(X) \cap \dom(h)$, we have $h(\ell) =_{\I(X)\cup S} h'(\ell)$, 
\item if $\card{\dom(h) \setminus \I(X)} < n$ then 
  $\card{\dom(h) \setminus \I(X)} = \card{\dom(h') \setminus \I(X)}$, 
\item if $\card{\dom(h) \setminus \I(X)} \geq n$ then 
  $\card{\dom(h') \setminus \I(X)} \geq n$. 
\end{compactenum}
\end{definition}
Observe that, for any $n \leq m$ and $S \subseteq T$ we have
$\sim^\I_{m,X,T} ~\subseteq~ \sim^\I_{n,X,S}$. In addition, for any
integer $k>0$, subset $S \subseteq U^\I$ and location $\ell \in U^\I$,
we consider the function $\mathit{prun}_{k,S}^\ell(\ell_1, \ldots,
\ell_k)$, which replaces each value $\ell_i \not\in S$ in its argument
list by $\ell$.

\begin{lemma}\cite[Lemma 94]{YangPhd}\label{lemma:small-model-existence}
Given an interpretation $\I$ and a heap $h : U^\I
\rightharpoonup_{\mathrm{fin}} (U^\I)^k$, for each integer $n>0$, each
set of variables $X \subseteq \vars$, each set of locations $L
\subseteq U^\I$ such that $L \cap \I(X) = \emptyset$ and $\card{L}=n$,
and each location $v \in U^\I \setminus (\I(X) \cup \{\nil^\I\} \cup
L)$, there exists a heap $h' : U^\I \rightharpoonup_{\mathrm{fin}}
(U^\I)^k$, with the following properties: \begin{compactenum}
\item $h \sim^\I_{n,X,L} h'$, 
\item $\dom(h') \setminus \I(X) \subseteq L$, 
\item for all $\ell \in \dom(h')$, we have $h'(\ell) =
  \mathit{prun}_{k,\I(X)\cup L}^v(h(\ell))$. 
\end{compactenum}
\end{lemma}
Next, we define the following measure on quantifier-free $\seplog$
formulae:
\[\begin{array}{ccccccc}
\len{\phi*\psi} = \len{\phi}+\len{\psi} & \hspace*{2mm} & 
\len{\phi \wand \psi} = \len{\psi} & \hspace*{2mm} & 
\len{\phi\wedge\psi} = \max(\len{\phi},\len{\psi}) & \hspace*{2mm} & 
\len{\neg\phi} = \len{\phi} \\
\len{\tterm \mapsto \uterm} = 1 & \hspace*{2mm} & 
\len{\emp} = 1 & \hspace*{2mm} & 
% \len{\forall x . \phi(x)} = \len{\phi} & \hspace*{2mm} & 
\len{\phi} = 0 \text{ if $\phi$ is a $\Sigma$-formula}
\end{array}\]
Intuitively, $\len{\varphi}$ is the maximum number of invisible
locations, that are not in $\I(\fv{\varphi})$, and which can be
distinguished by the quantifier-free $\seplog(\euf)_{U,U^k}$ formula
$\varphi$. The crux of the \textsc{PSPACE}-completeness proof for
quantifier-free $\seplog(\euf)_{U,U^k}$ is that two heaps equivalent
up to $\len{\varphi}$ invisible locations are also equivalent from the
point of view of satisfiability of $\varphi$, which provides a small
model property for this fragment \cite{YangPhd,CalcagnoYangOHearn01}.

\begin{lemma}\cite[Prop. 95]{YangPhd}\label{lemma:heap-equivalence}
Given a quantifier-free $\seplog(\euf)_{U,U^k}$ formula $\varphi$, an
interpretation $\I$, and two heaps $h$ and $h'$, if $h
\sim^\I_{\len{\varphi},\fv{\varphi},\emptyset} h'$ and $\I,h
\models_{\tinyseplog} \varphi$ then $\I,h' \models_{\tinyseplog}
\varphi$.
\end{lemma}
Our aim is to extend this result to
$\exists^*\forall^*\seplog(\euf)_{U,U^k}$, in the first place. This
new small model property is given by the next lemma. 

\begin{lemma}\label{lemma:small-model-property}
Let $\varphi(x_1^U,\ldots,x_n^U)$ be a quantifier-free
$\seplog(\euf)_{U,U^k}$-formula, and $\varphi^\forall \equiv \forall
x_1^U \ldots \\ \forall x_n^U ~.~ \varphi(x_1^U,\ldots,x_n^U)$ be its
universal closure. Then $\varphi^\forall$ has a model if and only if
there exists an interpretation $\I$ and a heap $h : U^\I
\rightharpoonup_{\mathrm{fin}} (U^\I)^k$ such that $\I,h
\models_{\tinyseplog} \varphi^\forall$ and: \begin{compactenum}
\item $\card{U^\I} \leq \len{\varphi} + \card{\fv{\varphi^\forall}} + n$, 
\item $\dom(h) \subseteq L \cup \I(\fv{\varphi^\forall})$, 
\item for all $\ell \in \dom(h)$, we have $h(\ell) \in
  (\I(\fv{\varphi^\forall}) \cup \{\nil^\I\} \cup L \cup \{v\})^k$,
\end{compactenum}
where $L \subseteq U^\I \setminus \I(\fv{\varphi^\forall})$ is a
set of locations such that $\card{L} = \len{\varphi}+n$ and $v \in
U^\I \setminus (\I(\fv{\varphi^\forall}) \cup \{\nil^\I\} \cup L)$
is an arbitrary location.
\end{lemma}
\proof{ 
``$\Rightarrow$'' Suppose that $\varphi^\forall$ has a model,
  i.e. $\I',h' \models_{\tinyseplog} \forall x_1 \ldots \forall
  x_n.\varphi(x_1, \ldots, x_n)$ for some interpretation $\I'$ and
  some heap $h' : U^{\I'} \rightharpoonup_{\mathrm{fin}} (U^{\I'})^k$. 
  We consider the case in which $\card{\dom(h') \setminus \I'(\fv{\varphi^\forall})} > \len{\varphi} + n$ 
  --- the other case $\card{\dom(h') \setminus \I'(\fv{\varphi^\forall})} \leq
  \len{\varphi} + n$ is an easy check left to the reader. Because
  $U^{\I'}$ is countable, we can choose a subset consisting of
  $\len{\varphi}+n$ locations from $\dom(h') \setminus
  \I'(\fv{\varphi^\forall})$, and let $L = \{\ell_1, \ldots, \ell_{\len{\varphi}+n}\}$ 
  be this set. By Lemma \ref{lemma:small-model-existence}, there exists a heap 
  $h : U^{\I'} \rightharpoonup_{\mathrm{fin}} (U^{\I'})^k$ such that: \begin{compactitem}
    \item $h \sim^{\I'}_{\len{\varphi}+n,\fv{\varphi^\forall},L} h'$, 
    \item $\dom(h) \setminus \I'(\fv{\varphi^\forall}) \subseteq L$ and 
    \item $h(\ell) = \mathit{prun}^v_{\I'(\fv{\varphi^\forall}) \cup L}(h'(\ell))$,
      for all $\ell \in \dom(h)$.
  \end{compactitem}
  We define $\I = \I'[U \leftarrow \dom(h) \cup \I'(\fv{\varphi^\forall})]$ and
  prove that $\I,h \models_{\tinyseplog} \varphi^\forall$. Clearly the pair $\I,h$
  satisfies the requirements from the statement of the lemma. We have
  to prove that $\I[x_1 \leftarrow u_1] \ldots [x_n\leftarrow u_n],h 
  \models_{\tinyseplog} \varphi(x_1, \ldots, x_n)$, for all 
  $u_1, \ldots, u_n \in U^\I = \dom(h) \cup \I'(\fv{\varphi^\forall})$. 
  By Lemma \ref{lemma:heap-equivalence}, it is sufficient to prove that:
  \[h \sim^{\I[x_1 \leftarrow u_1] \ldots [x_n \leftarrow u_n]}_{\len{\varphi},\fv{\varphi},\emptyset} h'
  \text{, for all } u_1, \ldots, u_n \in \dom(h) \cup
  \I'(\fv{\varphi^\forall})\] Since $\I$ and $\I'$ agree on all variables from
  $\vars$, and $\I'[x_1 \leftarrow u_1] \ldots [x_n \leftarrow
    u_n],h' \models_{\tinyseplog} \varphi(x_1, \ldots, x_n)$, for all
  $u_1, \ldots, u_n \in U^{\I'}$, by the hypothesis, we
  obtain $\I,h \models_{\tinyseplog} \forall x_1 \ldots \forall x_n ~.~
  \varphi(x)$. The proof is by induction on $n>0$.

  \vspace*{\baselineskip}
  \paragraph{The base case $n=1$.} Let us prove the requirements of Definition
  \ref{def:heap-equivalence}: \begin{compactenum}
  \item $\I[x_1 \leftarrow u_1](\fv{\varphi}) \cap \dom(h) =
    \I[x_1 \leftarrow u_1](\fv{\varphi}) \cap \dom(h')$: observe first
    that $\fv{\varphi} = \fv{\varphi^\forall} \cup \set{x_1}$, thus we have:
    \[\I[x_1 \leftarrow u_1](\fv{\varphi}) = \I(\fv{\varphi^\forall}) \cup
    \set{u_1} = \I'(\fv{\varphi^\forall}) \cup \set{u_1}\enspace.\] If
    $u_1 \in \dom(h) \setminus \I'(\fv{\varphi^\forall})$, then $u_1
    \in L$, because $\dom(h) \subseteq L$, and implicitly $u_1 \in
    \dom(h')$, since $L \subseteq \dom(h') \setminus
    \I'(\fv{\varphi^\forall})$. In this case, we have:
    \[\begin{array}{rcll}
    \I[x_1 \leftarrow u_1](\fv{\varphi}) \cap \dom(h) & = & 
    (\I'(\fv{\varphi^\forall}) \cup \set{u_1}) \cap \dom(h) \\
    & = & (\I'(\fv{\varphi^\forall}) \cap \dom(h)) \cup \set{u_1} \\
    & = & (\I'(\fv{\varphi^\forall}) \cap \dom(h')) \cup \set{u_1} & 
    \text{since $h \sim^{\I'}_{\len{\varphi}+n,\fv{\varphi^\forall},L} h'$} \\
    & = & \I[x_1 \leftarrow u_1](\fv{\varphi}) \cap \dom(h') & \enspace.
    \end{array}\]
    On the other hand, if $u_1 \in \I'(\fv{\varphi^\forall})$, we have
    $\I[x_1 \leftarrow u_1](\fv{\varphi}) = \I'(\fv{\varphi^\forall})$
    and the result follows immediately.
  \item $h(\ell') =_{\I[x_1 \leftarrow u_1](\fv{\varphi})} h'(\ell')$,
    for all $\ell' \in \I[x_1 \leftarrow u_1](\fv{\varphi}) \cap
    \dom(h)$: by the definition of $h$ (Lemma
    \ref{lemma:small-model-existence}), we have $h(\ell') =
    \mathit{prun}^v_{\I'(\fv{\varphi^\forall}) \cup L}(h'(\ell'))$,
    for all $\ell' \in \dom(h)$. Hence we have $h(\ell')
    =_{\I'(\fv{\varphi^\forall}) \cup L} h'(\ell')$ and, consequently
    $h(\ell') =_{\I'(\fv{\varphi^\forall}) \cup \set{u_1}} h'(\ell')$,
    for all $\ell' \in \dom(h)$, since $u_1 \in
    \I'(\fv{\varphi^\forall}) \cup L$.
  \item $\card{\dom(h')\setminus\I[x_1 \leftarrow
      u_1](\fv{\varphi})} = \card{\dom(h') \setminus
    (\I(\fv{\varphi^\forall}) \cup \set{u_1})} \geq \card{\dom(h')
    \setminus \I(\fv{\varphi^\forall})} - 1 > \len{\varphi}$, by the
    previous assumption. Since $h'
    \sim^{\I}_{\len{\varphi}+1,\fv{\varphi^\forall},L} h$, we get
    $\card{\dom(h) \setminus \I(\fv{\varphi^\forall})} \geq
    \len{\varphi}+1$, thus $\card{\dom(h) \setminus \I[x_1 \leftarrow
        u_1](\fv{\varphi^\forall})} = \card{\dom(h) \setminus
      (\I(\fv{\varphi^\forall}) \cup \set{u_1})} \geq \len{\varphi}$.
  \end{compactenum}

  \vspace*{\baselineskip}
  \paragraph{The induction step $n>1$.} 
  We prove the points of Definition \ref{def:heap-equivalence},
  similar to the base case: \begin{compactenum}
  \item $\I[x_1 \leftarrow u_1] \ldots [x_n \leftarrow
    u_n](\fv{\varphi}) \cap \dom(h) = I[x_1 \leftarrow u_1] \ldots
    [x_n \leftarrow u_n](\fv{\varphi}) \cap \dom(h')$: we distinguish
    the case\begin{inparaenum}[(i)]
      \item $u_1 \in \dom(h) \setminus \I[x_2 \leftarrow u_2]
        \ldots [x_n \leftarrow u_n](\fv{\varphi})$ from
  \item $u_1 \in \I[x_2 \leftarrow u_2] \ldots [x_n \leftarrow
    u_n](\fv{\varphi})$.
    \end{inparaenum} In the first case, we have:
    \[\begin{array}{rcll}
    \I[x_1 \leftarrow u_1] \ldots [x_n \leftarrow u_n](\fv{\varphi}) \cap \dom(h) & = & 
    (\I[x_2 \leftarrow u_2] \ldots [x_n \leftarrow u_n](\fv{\varphi}) \cup \set{u_1}) \cap \dom(h) \\
    & = & (\I[x_2 \leftarrow u_2] \ldots [x_n \leftarrow u_n](\fv{\varphi}) \cap \dom(h)) \cap \set{u_1} \\
    \text{by the induction hypothesis } & = & (\I[x_2 \leftarrow u_2] \ldots [x_n \leftarrow u_n](\fv{\varphi}) \cap \dom(h')) \cap \set{u_1} \\
    & = & \I[x_1 \leftarrow u_1] \ldots [x_n \leftarrow u_n](\fv{\varphi}) \cap \dom(h') & \enspace.
    \end{array}\]
    If $u_1 \in \I[x_2 \leftarrow u_2] \ldots [x_n \leftarrow
      u_n](\fv{\varphi})$, we have $\I[x_1 \leftarrow u_1] \ldots [x_n
      \leftarrow u_n](\fv{\varphi}) = \I[x_2 \leftarrow u_2] \ldots
    [x_n \leftarrow u_n](\fv{\varphi})$ and an application of the
    induction hypothesis concludes the proof.
  \item By the construction of $h$, we have $h(\ell)
    =_{\I(\fv{\varphi^\forall}) \cup L} h'(\ell)$, for all $\ell \in
    \dom(h)$, thus $h(\ell) =_{\I[x_1\leftarrow
        u_1]\ldots[x_n\leftarrow u_n](\fv{\varphi})} h'(\ell)$, for
    all $\ell \in dom(h) \cap \I[x_1\leftarrow
      u_1]\ldots[x_n\leftarrow u_n](\fv{\varphi})$.
  \item Similar to the base case, we have: 
    \[\begin{array}{rcl}
    \card{\dom(h')\setminus\I[x_1 \leftarrow u_1] \ldots [x_n
        \leftarrow u_n](\fv{\varphi})} & = & \card{\dom(h') \setminus
      (\I(\fv{\varphi^\forall}) \cup \set{u_1, \ldots, u_n})} \\ &
    \geq & \card{\dom(h') \setminus \I(\fv{\varphi^\forall})} - n \\ &
    > & \len{\varphi} \text{, and}
    \\ \card{\dom(h)\setminus\I[x_1 \leftarrow u_1] \ldots [x_n
        \leftarrow u_n](\fv{\varphi})} & = & \card{\dom(h') \setminus
      (\I(\fv{\varphi^\forall}) \cup \set{u_1, \ldots, u_n})}
    \\ & \geq & \len{\varphi}\enspace.
    \end{array}\]
  \end{compactenum}

  Since the direction ``$\Leftarrow$'' is trivial, this concludes our
  proof. \qed
}

We are ready to prove two decidability results, based on the above
small model property, concerning the cases where\begin{inparaenum}[(i)] 
\item $\locs$ is interpreted as a countable set with equality, and
\item $\locs$ is interpreted as an infinite countable set with no
  other operators than equality.
\end{inparaenum}

\subsection{Uninterpreted Locations without Cardinality Constraints}

In this section, we consider the satisfiability problem for the
fragment $\exists^*\forall^*\seplog(\euf)_{U,U^k}$, where the location
sort $U$ can be interpreted by any (possibly finite) countable set,
with no other operations than the equality, and the data sort consists
of $k$-tuples of locations. 

\begin{theorem}\label{thm:ul-bsr}
  The satisfiability problem for
  $\exists^*\forall^*\seplog(\euf)_{U,U^k}$ is
  \textsc{PSPACE}-complete.
\end{theorem}
\proof{ \textsc{PSPACE}-hardness follows from the fact that
  satisfiability is \textsc{PSPACE}-complete for quantifier-free
  $\seplog(\euf)_{U,U^k}$ \cite{CalcagnoYangOHearn01}. To prove
  membership in \textsc{PSPACE}, consider the formula $\phi \equiv
  \exists x_1 \ldots \exists x_m \forall y_1 \ldots \forall y_n ~.~
  \varphi(\vec{x}, \vec{y})$, where $\varphi$ is a quantifier-free
  $\seplog(\euf)_{U,U^k}$ formula. Let $\vec{c} = \tuple{c_1, \ldots,
    c_m}$ be a tuple of constant symbols, and $\widetilde{\phi} \equiv
  \forall y_1 \ldots \forall y_n ~.~ \varphi(\vec{c},\vec{y})$ be the
  functional form of $\phi$, obtained by replacing $x_i$ with $c_i$,
  for all $i=1,\ldots,m$. By Lemma \ref{lemma:small-model-property},
  $\widetilde{\phi}$ has a model if and only if it has a model $\I,h$
  such that: \begin{compactitem}
   \item $\card{U^\I}\leq \len{\varphi}+n+m$, 
   \item $\dom(h) \subseteq L \cup \vec{c}^\I$, 
   \item $\forall \ell \in \dom(h) ~.~ h(\ell) \in (\I(\vec{c}) \cup
     \{\nil^\I\} \cup L \cup \set{v})^k$,
  \end{compactitem}
  where $L \subseteq U^\I \setminus \I(\vec{c})$, $\card{L} =
  \len{\varphi}+m$ and $v \in U^\I \setminus (\I(\vec{c}) \cup
  \{\nil^\I\} \cup L)$. We describe below a nondeterministic
  polynomial space algorithm that decides satisfiability of
  $\widetilde{\phi}$. First, nondeterministically chose a model $\I,h$
  that meets the above requirements. Then we check, for each tuple
  $\tuple{u_1, \ldots, u_n} \in (U^\I)^n$ that $\I[y_1\leftarrow
    u_1]\ldots[y_n\leftarrow u_n],h \models_{\tinyseplog} \varphi$. In
  order to enumerate all tuples from $(U^\I)^n$ we need
  $n\cdot\lceil\log_2(\len{\varphi}+n+m)\rceil$ extra bits, and the check  
  for each such tuple can be done in \textsc{PSPACE}, according to
  \cite[\S5]{CalcagnoYangOHearn01}. \qed}

This result is somewhat surprising, because the classical
Bernays-Sch\"onfinkel fragment of first-order formulae with predicate
symbols (but no function symbols) and quantifier prefix
$\exists^*\forall^*$ is known to be \textsc{NEXPTIME}-complete
\cite[\S7]{Lewis80}. The explanation lies in the fact that the
interpretation of an arbitrary predicate symbol
$P(\xterm_1,\ldots,\xterm_n)$ cannot be captured using only points-to
atomic propositions, e.g.\ $\xterm_1 \mapsto (\xterm_2, \ldots,
\xterm_n)$, between locations and tuples of locations, due to the
interpretation of points-to's as heaps\footnote{If $\xterm_1 \mapsto
  (\xterm_2, \ldots, \xterm_n)$ and $\xterm_1 \mapsto (\xterm'_2,
  \ldots, \xterm'_n)$ hold, this forces $\xterm_i=\xterm'_i$, for all
  $i=2,\ldots,n$.} (finite partial functions).

The following lemma sets a first decidability boundary for
$\seplog(\euf)_{U,U^k}$, by showing how extending the quantifier
prefix to $\exists^*\forall^*\exists^*$ leads to undecidability. 

\begin{lemma}\label{lemma:exists-forall-exists-undecidability}
  The satisfiability problem for
  $\exists^*\forall^*\exists^*\seplog(\euf)_{U,U^k}$ is undecidable,
  if $k\geq2$.
\end{lemma}
\proof{  
  By reduction from the undecidability of the following tiling
  problem. Let $\mathcal{T}=\set{T_0, \ldots, T_s}$ be a set of tile
  types and $H,V \subseteq \mathcal{T} \times \mathcal{T}$ be two
  relations between tile types. Given $n,m>1$, a tiling of the
  $n\times m$ square is a function $\tau : [1,n] \times [1,m]
  \rightarrow \mathcal{T}$ such that: \begin{compactitem}
  \item $\tau(1,1) = T_0$, 
  \item $(\tau(i,j),\tau(i+1,j)) \in H$, for all $i \in [1,n-1]$ and $j \in [1,m]$, 
  \item $(\tau(i,j),\tau(i,j+1)) \in V$, for all $i \in [1,n]$ and $j \in [1,m-1]$. 
  \end{compactitem}
  The existence of $n,m>1$ and of a tiling of the $n\times m$ square
  is a well-known undecidable problem. We encode this as the
  satisfiability of a formula in
  $\exists^*\forall^*\exists^*\seplog(\euf)_{U,U^k}$. Let $\ell_0,
  \ell_1, t_0, \ldots, t_s$ be constant symbols of sort $U$.

  The basic idea is to represent each cell of the $n\times m$ grid by
  a heap cell, defined by the atomic proposition $\xterm \mapsto
  (b,r,\yterm,\zterm,T)$, where: \begin{compactitem}
  \item $b,r \in \set{\ell_0,\ell_1}$ act as binary flags indicating
    whether the cell belongs to the bottom $n$ row ($b=\ell_1$), and
    the righmost $m$ column ($r=\ell_1$), respectively,
  \item$\yterm$ and $\zterm$ are the horizontal (right) and vertical
    (below) successors of $\xterm$, and
  \item $T \in \set{t_0,\ldots,t_s}$ is the type of the tile covering
    the cell pointed to by $\xterm$.
  \end{compactitem}
  The finite relations $H,V \subseteq \mathcal{T} \times \mathcal{T}$
  are defined by the finite disjunctions $h(x,y),v(x,y)$ of equalities
  involving $t_0, \ldots, t_s$. We write $\xterm \hookrightarrow
  (b,r,\yterm,\zterm,T)$ for $\xterm \mapsto (b,r,\yterm,\zterm,T) *
  \top$, $\xterm \hookrightarrow (\_~,r,\yterm,\zterm,T)$ for
  $\bigvee_{b\in\set{\ell_0,\ell_1}} \xterm \hookrightarrow
  (b,r,\yterm,\zterm,T)$, $\xterm \hookrightarrow
  (b,\_~,\yterm,\zterm,T)$ for $\bigvee_{r\in\set{\ell_0,\ell_1}}
  \xterm \hookrightarrow (b,r,\yterm,\zterm,T)$, $\xterm
  \hookrightarrow (b,r,\yterm,\zterm,\_)$ for $\bigvee_{T \in
    \set{t_0,\ldots,t_s}} \xterm \hookrightarrow
  (b,r,\yterm,\zterm,T)$, and $\xterm \hookrightarrow (b,r,\_~,\_~,T)$
  for $\exists \yterm \exists \zterm ~.~ \xterm \hookrightarrow
  (b,r,\yterm,\zterm,T)$. We define the formula $\Phi$ as the
  conjunction of the formulae below, with the pairwise disequality
  constraint $\bigwedge_{c,c' \in \set{\ell_0, \ell_1, t_0, \ldots,
      t_s}}c \not\teq c'$:
  \begin{eqnarray}
  \exists \uterm \forall \yterm \forall \zterm ~.~ \uterm \hookrightarrow (0,0,\_~,\_~,t_0) \wedge \neg \yterm \hookrightarrow (\_~,\_~,\uterm,\_~,\_) \wedge \neg \zterm \hookrightarrow (\_~,\_~,\_~,\uterm,\_) \label{eq:start} \\
  \nonumber \\[-2mm]
  \forall \xterm \forall \yterm \forall \zterm ~.~ \bigwedge_{b\in\set{\ell_0,\ell_1}} \left[\xterm \hookrightarrow (b,0,\yterm,\zterm,\_) \Rightarrow \yterm \hookrightarrow (b,\_~,\_~,\_~,\_) \wedge \xterm \not\teq \yterm \right] \label{eq:not-right} \\
  \forall \xterm \forall \yterm \forall \zterm ~.~ \xterm \hookrightarrow (\_~,1,\yterm,\zterm,\_) \Rightarrow \xterm\teq\yterm \label{eq:right} \\
  \nonumber \\[-2mm]
  \forall \xterm \forall \yterm \forall \zterm ~.~ \bigwedge_{r\in\set{\ell_0,\ell_1}} \left[\xterm \hookrightarrow (0,r,\yterm,\zterm,\_) \Rightarrow \yterm \hookrightarrow (\_~,r,\_~,\_~,\_) \wedge \xterm \not\teq \zterm \right] \label{eq:not-bot} \\
  \forall \xterm \forall \yterm \forall \zterm ~.~ \xterm \hookrightarrow (1,\_~,\yterm,\zterm,\_) \Rightarrow \xterm\teq\zterm \label{eq:bot} \\
  \nonumber \\[-2mm]
  \forall \xterm \forall \yterm \forall \yterm' \forall \yterm'' \forall \zterm \forall \zterm' \forall \zterm'' ~.~ \xterm \hookrightarrow (\_~,\_~,\yterm,\zterm,\_) * \yterm \mapsto (\_~,\_~,\yterm',\zterm',\_) * \zterm \mapsto (\_~,\_~,\yterm'',\zterm'',\_) \nonumber \\ 
  \Rightarrow \zterm'\teq\yterm'' \label{eq:grid} \\
  \nonumber \\[-2mm]
  \forall \xterm \forall \yterm \forall \zterm ~.~ \bigwedge_{T,T' \in \set{t_0,\ldots,t_s}} \xterm \hookrightarrow (\_~,\_~,\yterm,\zterm,T) * \yterm \mapsto (\_~,\_~,\_~,\_~,T') \Rightarrow h(T,T') \label{eq:hor} \\
  \forall \xterm \forall \yterm \forall \zterm ~.~ \bigwedge_{T,T' \in \set{t_0,\ldots,t_s}} \xterm \hookrightarrow (\_~,\_~,\yterm,\zterm,T) * \zterm \mapsto (\_~,\_~,\_~,\_~,T') \Rightarrow v(T,T') \label{eq:ver}
  \end{eqnarray}
  The intuition of these formulae is as follows: (\ref{eq:start}) is
  the initial constraint, asking that the top-left corner is labeled
  with $T_0$, (\ref{eq:not-right}) requires that each cell not on the
  rightmost column has a distinct left successor, (\ref{eq:right}) is
  the dual constraint, asking that the left successor of each cell on
  the rightmost column is the cell itself, (\ref{eq:not-bot}) and
  (\ref{eq:bot}) are the similar constraints for the bottom
  successors, (\ref{eq:grid}) is the grid constraint, and
  (\ref{eq:hor}), (\ref{eq:ver}) are the horizontal and vertical
  constraints on the types of tiles. Observe that the quantifier
  prefix of the formulae (\ref{eq:start}), (\ref{eq:not-right}) and
  (\ref{eq:not-bot}) belongs necessarily to the language produced by
  the regular expression $\exists^*\forall^*\exists^*$.

  Observe that the formulae (\ref{eq:start} - \ref{eq:ver}) use $k=5$
  record fields. We can reduce the value of $k$ to $2$, by using the
  following encoding of the atomic propositions in $\Phi$:
  \[\begin{array}{rcl}
  \xterm \hookrightarrow (b,r,\yterm,\zterm,T) & \equiv & \exists \xterm_0
  \exists \xterm_1 \exists \xterm_2 \exists \xterm_3 ~.~ \xterm \hookrightarrow
  (\xterm_0,\xterm_1) * \xterm_1 \mapsto (\yterm,\zterm) * \\ 
  && \xterm_0 \mapsto (\xterm_2, b) * \xterm_2 \mapsto (\xterm_3, r) * \xterm_3 \mapsto (T,\nil)
  \end{array}\]
  It is easy to show now that $\Phi$ has a model iff there exists
  $n,m>1$ and a tiling of the $n\times m$ grid, which proves the
  undecidability of the satisfiability problem for
  $\exists^*\forall^*\exists^*\seplog(\euf)_{U,U^k}$, when
  $k\geq2$. \qed
}

Observe that the result of Lemma
\ref{lemma:exists-forall-exists-undecidability} sets a fairly tight
boundary between the decidable and undecidable fragments of
$\seplog$. On the one hand, simplifying the quantifier prefix to
$\exists^*\forall^*$ yields a decidable fragment (Theorem
\ref{thm:ul-bsr}), whereas $\seplog(\euf)_{U,U}$ ($k=1$) without the
magic wand ($\wand$) is decidable with non-elementary time complexity,
even when considering an unrestricted quantifier prefix
\cite{BrocheninDemriLozes11}.

\subsection{Uninterpreted Locations with Cardinality $\aleph_0$}

We consider the stronger version of the satisfiability problem for
$\exists^*\forall^*\seplog(\euf)_{U,U^k}$, where $U$ is interpreted as
an infinite countable set (of cardinality $\aleph_0$) with no function
symbols, other than equality. Instances of this problem occur when,
for instance, the location sort is taken to be $\Int$, but no
operations are used on integers, except for testing equality.

Observe that this restriction changes the satisfiability status of
certain formulae. For instance, $\exists \xterm \forall \yterm ~.~
\yterm \not\teq \nil \Rightarrow (\yterm \mapsto \xterm * \top )$ is
satisfiable if $U$ is interpreted as a finite set, but becomes
unsatisfiable when $U$ is infinite. The reason is that this formula
requires every location from $U^\I$ apart from $\nil$ to be part of
the domain of the heap, which is impossible due the fact that only
finite heaps are considered by the semantics of $\seplog$.

In the following proof, we use the formula $\mathsf{alloc}(\xterm)
\equiv \xterm \mapsto (\xterm,\ldots,\xterm) \wand \bot$, expressing
the fact that a location variable $\xterm$ is \emph{allocated},
i.e.\ its interpretation is part of the heap's domain
\cite{BrocheninDemriLozes11}. Intuitively, we reduce any instance of
the $\exists^*\forall^*\seplog(\euf)_{U,U^k}$ satisfiability problem,
with $U$ of cardinality $\aleph_0$, to an instance of the same problem
without this restriction, by the following cut-off argument: if a free
variable is interpreted as a location which is neither part of the
heap's domain, nor equal to the interpretation of some constant, then
it is not important which particular location is chosen for that
interpretation.

\begin{theorem}\label{thm:ul-aleph-zero-bsr}
  The satisfiability problem for
  $\exists^*\forall^*\seplog(\euf)_{U,U^k}$ is
  \textsc{PSPACE}-complete if $U$ is required to have cardinality
  $\aleph_0$.
\end{theorem}
\proof{ \textsc{PSPACE}-hardness follows from the
  \textsc{PSPACE}-completeness of the satisfiability problem for
  quantifier-free $\seplog$, with uninterpreted locations
  \cite[\S5.2]{CalcagnoYangOHearn01}. Since the reduction from
  \cite[\S5.2]{CalcagnoYangOHearn01} involves no universally quantified
  variables, the $\aleph_0$ cardinality constraint has no impact on
  this result.
 
  Let $\exists x_1 \ldots \exists x_m \forall y_1 \ldots \forall y_n
  ~.~ \varphi(\vec{x},\vec{y})$ be a formula, and $\forall y_1 \ldots
  \forall y_n ~.~ \varphi(\vec{c},\vec{y})$ be its functional form,
  obtained by replacing each $x_i$ with $c_i$, for $i=1,\ldots,m$. We
  consider the following formulae, parameterized by $y_i$, for $i=1,\ldots,n$:
  \[\begin{array}{rcl}
  \psi_0(y_i) & \equiv & \mathsf{alloc}(y_i) \\
  \psi_1(y_i) & \equiv & \bigvee_{j=1}^m y_i=c_j \\
  \psi_2(y_i) & \equiv & y_i=d_i \\
  \mathsf{external} & \equiv & \bigwedge_{i=1}^n 
  (\neg\mathsf{alloc}(d_i) \wedge \bigwedge_{j=1}^m d_i \neq c_j)
  \end{array}\]
  where $\{d_i \mid i = 1, \ldots, n\}$ is a set of fresh constant
  symbols.
  % not occurring in $\forall y_1 \ldots \forall y_n ~.~ \varphi(\vec{c},\vec{y})$. 
  We show the following fact:

  \begin{fact}\label{fact:aleph0-equisat}
    There exists an interpretation $\I$ and a heap $h$ such that
    $\card{U^\I}=\aleph_0$ and $\I,h \models_{\tinyseplog} \forall y_1
    \ldots \forall y_n ~.~ \varphi(\vec{c},\vec{y})$ iff there exists
    an interpretation $\I'$, not constraining the cardinality of
    $U^{\I'}$, and a heap $h'$ such that:
    \[\I',h' \models_{\tinyseplog} \mathsf{external} \wedge \forall y_1 \ldots \forall y_n 
    \bigwedge_{\tuple{t_1,\ldots,t_n} \in \set{0,1,2}^n}
    \underbrace{\bigwedge_{i=1}^n \left(\psi_{t_i}(y_i) \Rightarrow
      \varphi(\vec{c},\vec{y})\right)}_{\Psi_{\tuple{t_1,\ldots,t_n}}}\]
  \end{fact}
  \proof{
    ``$\Rightarrow$'' This direction is immediate, because for
    each location $\ell \in U^\I$, we have $\I[y\leftarrow\ell],h
    \models_{\tinyseplog} \psi_0(y) \vee \psi_1(y) \vee \psi_2(y)$ and
    $\I[y_1 \leftarrow \ell_1] \ldots [y_n \leftarrow \ell_n],h
    \models_{\tinyseplog} \varphi(\vec{c},\vec{y})$, for all
    $\tuple{\ell_1,\ldots,\ell_n} \in (U^\I)^n$. ``$\Leftarrow$''
    Because the cardinality of $U^{\I'}$ is unconstrained, by Lemma
    \ref{lemma:small-model-property}, there exists an interpretation
    $\I''$ and a heap $h$ such that
    $\I''=\I'[U\leftarrow\dom(h)\cup\I'(\fv{\Psi})]$ and $\I'',h
    \models_{\tinyseplog} \mathsf{extern} \wedge \forall y_1 \ldots
    \forall y_n \bigwedge_{\tuple{t_1,\ldots,t_n} \in \set{0,1,2}^n}
    \Psi_{\tuple{t_1,\ldots,t_n}}$. We obtain that
    $\I''[y_1\leftarrow\ell_1]\ldots[y_n\leftarrow\ell_n]
    \models_{\tinyseplog} \varphi(\vec{c},\vec{y})$ for each tuple
    $\tuple{\ell_1,\ldots,\ell_n} \in (U^{\I''})^n$ where either
    \begin{inparaenum}[(i)]
      \item $\ell_i \in \dom(h)$, 
      \item $\ell_i = \I''(c_j)$ for some $j=1,\ldots,m$, or
      \item $\ell_i = \I''(d_j)$ for some $j=1,\ldots,n$ and neither
        of the previous hold.
    \end{inparaenum}
    Because it is not important which location is used for the
    interpretation of $d_j$, $j=1,\ldots,n$, we have $\I,h
    \models_{\tinyseplog} \forall y_1 \ldots \forall y_n ~.~
    \varphi(\vec{c},\vec{y})$ for every extension $\I$ of $\I''$ such
    that $\card{U^{\I}}=\aleph_0$. \qed
  }

  To show membership in \textsc{PSPACE}, consider a nondeterministic
  algorithm that choses $\I'$ and $h'$ and uses $2n$ extra bits to
  check that $\I',h' \models_{\tinyseplog} \mathsf{extern} \wedge
  \forall y_1 \ldots \forall y_n ~.~ \Psi_{\tuple{t_1,\ldots,t_n}}$
  separately, for each $\tuple{t_1,\ldots,t_n} \in \set{0,1,2}^n$. By
  Lemma \ref{lemma:small-model-property}, the sizes of $\I'$ and $h'$
  are bounded by a polynomial in the size of
  $\Psi_{\tuple{t_1,\ldots,t_n}}$, which is polynomial in the size of
  $\varphi$, and by Theorem \ref{thm:ul-bsr}, each of these checks can
  be done in polynomial space. \qed}

\subsection{Integer Locations with Linear Arithmetic}\label{sec:undecidability}

In the rest of this section we show that the
Bernays-Sch\"onfinkel-Ramsey fragment of $\seplog$ becomes undecidable
as soon as we use integers to represent the set of locations and
combine $\seplog$ with linear integer arithmetic ($\lia$). The proof
relies on an undecidability argument for a fragment of Presburger
arithmetic with one monadic predicate symbol, interpreted over finite
sets. Formally, we denote by $(\exists^*\forall^* \cap
\forall^*\exists^*)-\lia$ the set of formulae consisting of a
conjunction between two linear arithmetic formulae, one with
quantifier prefix in the language $\exists^*\forall^*$, and another
with quantifier prefix $\forall^*\exists^*$.

\begin{theorem}\label{thm:bsr-presburger}
  The satisfiability problem is undecidable for the fragment
  $(\exists^*\forall^* \cap \forall^*\exists^*)-\lia$, with one monadic
  predicate symbol, interpreted over finite sets of integers.
\end{theorem}
\proof{
 We reduce from the following variant of \emph{Hilbert's 10th
   Problem}: given a multivariate Diophantine polynomial $R(x_1,
 \ldots, x_n)$, the problem ``does $R(x_1, \ldots, x_n) = 0$ have
 a solution in $\nat^n$ ?'' is undecidable \cite{Matyiasevich70}.

 By introducing sufficiently many free variables, we encode 
 $R(x_1, \ldots, x_n) = 0$ as an equisatisfiable Diophantine
 system of degree at most two, containing only equations of the form
 $x=yz$ (resp. $x=y^2$) and linear equations $\sum_{i=1}^k a_ix_i =
 b$, where $a_1,\ldots,a_k,b\in\zed$. Next, we replace each equation
 of the form $x=yz$, with $y$ and $z$ distinct variables, with the
 quadratic system $2x+t_y+t_z=t_{y+z} \wedge t_y=y^2 \wedge t_z=z^2
 \wedge t_{y+z} = (y+z)^2$, where $t_y,t_z$ and $t_{y+z}$ are fresh
 (free) variables. In this way, we replace all multiplications between
 distinct variables by occurrences of the squaring function. Let
 $\Psi_{R(x_1,\ldots,x_n)=0}$ be the conjunction of the above
 equations. It is manifest that $R(x_1, \ldots, x_n) = 0$ has a
 solution in $\nat^n$ iff $\Psi_{R(x_1,\ldots,x_n)=0}$ is satisfiable,
 with all free variables ranging over $\nat$.

Now we introduce a monadic predicate symbol $P$, which is intended to
denote a (possibly finite) set of consecutive perfect squares,
starting with $0$. To capture this definition, we require the
following:
\[\begin{array}{l}\tag{$\mathsf{sqr}$}\label{eq:sqr}
P(0) \wedge P(1) \wedge \forall x \forall y \forall z ~.~ 
P(x) \wedge P(y) \wedge P(z) \wedge x < y < z ~\wedge \\
(\forall u ~.~ x < u < y \vee y < u < z \Rightarrow \neg P(u)) 
\Rightarrow z-y = y-x+2
\end{array}\]
Observe that this formula is a weakening of the definition of the
infinite set of perfect squares given by Halpern \cite{Halpern91},
from which the conjunct $\forall x \exists y ~.~ y>x \wedge P(y)$,
requiring that $P$ is an infinite set of natural numbers, has been
dropped. Moreover, notice that \ref{eq:sqr} has quantifier prefix
$\forall^3\exists$, due to the fact that $\forall u$ occurs implicitly under
negation, on the left-hand side of an implication. If $P$ is
interpreted as a finite set $P^\I = \set{p_0, p_1, \ldots, p_N}$ such
that (w.l.o.g.) $p_0 < p_1 < \ldots < p_N$, it is easy to show, by
induction on $N>0$, that $p_i = i^2$, for all $i=0,1,\ldots,N$.

The next step is encoding the squaring function using the monadic
predicate $P$. This is done by replacing each atomic proposition
$x=y^2$ in $\Psi_{R(x_1,\ldots,x_n)=0}$ by the formula $\theta_{x=y^2}
\equiv P(x) \wedge P(x+2y+1) \wedge \forall z ~.~ x<z<x+2y+1
\Rightarrow \neg P(z)$. %Let us now prove the following fact:

\begin{fact}\label{fact:theta}
For each interpretation $\I$ mapping $x$ and $y$ into $\nat$, $\I
\models x=y^2$ iff $\I$ can be extended to an interpretation of $P$ as
a finite set of consecutive perfect squares such that $\I \models
\theta_{x=y^2}$.
\end{fact}
\proof{
  ``$\Rightarrow$'' If $\I \models x=y^2$ then we have
  $(y^\I+1)^2 = x^\I + 2y^\I + 1$. Let $P^\I$ be the set
  $\{0,1,\ldots,(y^\I+1)^2\}$. Clearly $x^\I, x^\I + 2y^\I + 1 \in
  P^\I$ and, since they are consecutive perfect squares, every number
  in between $x^\I$ and $x^\I + 2y^\I + 1$ does not belong to
  $P^\I$. Thus $\I \models \theta_{x=y^2}$. ``$\Leftarrow$'' If $\I
  \models \theta_{x=y^2}$ and $P^\I$ is a set of consecutive perfect
  squares, it follows that $x^\I$ and $x^\I+2y^\I+1$ are consecutive
  perfect squares, i.e.\ $x^\I=n^2$ and $x^\I+2y^\I+1=(n+1)^2$ for
  some $n \in \nat$. Then $y^\I=n$, thus $\I \models x=y^2$. \qed
}

Let $\Phi_{R(x_1,\ldots,x_n)=0}$ be the conjunction of \ref{eq:sqr}
with the formula obtained by replacing each atomic proposition $x=y^2$
with $\theta_{x=y^2}$ in $\Psi_{R(x_1,\ldots,x_n)=0}$. Observe that
each universally quantified variable in $\Phi_{R(x_1,\ldots,x_n)=0}$
occurs either in \ref{eq:sqr} or in some $\theta_{x=y^2}$, and
moreover, each $\theta_{x=y^2}$ belongs to the $\exists^*\forall^*$
fragment of $\lia$. $\Phi_{R(x_1,\ldots,x_n)=0}$ belongs thus to the
$\exists^*\forall^* \cap \forall^*\exists^*$ fragment of $\lia$, with
$P$ being the only monadic predicate symbol. Finally, we prove that
$R(x_1,\ldots,x_n)=0$ has a solution in $\nat^n$ iff
$\Phi_{R(x_1,\ldots,x_n)=0}$ is satisfiable.

``$\Rightarrow$'' Let $\I$ be a valuation mapping $x_1, \ldots, x_n$
into $\nat$, such that $\I \models R(x_1,\ldots,x_n)=0$. Obviously,
$\I$ can be extended to a model of $\Psi_{R(x_1,\ldots,x_n)=0}$ by
assigning $t_x^\I = (x^\I)^2$ for all auxiliary variables $t_x$
occurring in $\Psi_{R(x_1,\ldots,x_n)=0}$. We extend $\I$ to a model
of $\Phi_{R(x_1,\ldots,x_n)=0}$ by assigning $P^\I = \{n^2 \mid 0 \leq
n \leq \sqrt{m}\}$, where $m = \max\{(x^\I+1)^2 \mid x \in
\fv{\Psi_{R(x_1,\ldots,x_n)=0}}\}$. Clearly $P^\I$ meets the
requirements of \ref{eq:sqr}. By Fact \ref{fact:theta}, we obtain that
$\I \models \theta_{x=y^2}$ for each subformula $\theta_{x=y^2}$ of
$\Phi_{R(x_1,\ldots,x_n)=0}$, thus $\I \models
\Phi_{R(x_1,\ldots,x_n)=0}$.

``$\Leftarrow$'' If $\I \models \Psi_{R(x_1,\ldots,x_n)=0}$ then, by
\ref{eq:sqr}, $P^\I$ is a set of consecutive perfect squares, and,
by Fact \ref{fact:theta}, $\I \models x=y^2$ for each subformula
$\theta_{x=y^2}$ of $\Phi_{R(x_1,\ldots,x_n)=0}$. Then $\I \models
\Psi_{R(x_1,\ldots,x_n)=0}$ and consequently $\I \models
R(x_1,\ldots,x_n)=0$. \hfill\qed}

We consider now the satisfiability problem for the fragment
$\exists^*\forall^*\seplog(\lia)_{\Int,\Int}$ where both $\locs$ and
$\data$ are taken to be the $\Int$ sort, equipped with addition and
total order. Observe that, in this case, the heap consists of a set of
lists, possibly with aliases and circularities. Without losing
generality, we consider that $\Int$ is interpreted as the set of
positive integers\footnote{Extending the interpretation of $\locs$ to
  include negative integers does not make any difference for the
  undecidability result.}.

The above theorem cannot be directly used for the undecidability of
$\exists^*\forall^*\seplog(\lia)_{\Int,\Int}$, by interpreting the
(unique) monadic predicate as the (finite) domain of the heap. The
problem is with the \ref{eq:sqr} formula, that defines the
interpretation of the monadic predicate as a set of consecutive
perfect squares $0,1,\ldots,n^2$, and whose quantifier prefix lies in
the $\forall^*\exists^*$ fragment. We overcome this problem by
replacing the \ref{eq:sqr} formula above with a definition of such
sets in $\exists^*\forall^*\seplog(\lia)_{\Int,\Int}$. Let us first
consider the following properties expressed in $\seplog$
\cite{BrocheninDemriLozes11}: \[\begin{array}{rcl}
% \mathsf{alloc}(\xterm) & \equiv & \xterm \mapsto \xterm \wand \top \\ 
\sharp\xterm\geq1 & \equiv & \exists\uterm ~.~ \uterm \mapsto \xterm * \top \\ 
% \sharp\xterm\leq0 & \equiv & \neg(\sharp\xterm\geq1) \\ 
\sharp\xterm\leq1 & \equiv & \forall\uterm\forall\tterm ~.~ \neg(\uterm \mapsto \xterm * \tterm \mapsto \xterm * \top)
\end{array}\]
Intuitively, $\sharp\xterm\geq1$ states that $\xterm$ has at least one
predecessor in the heap, whereas $\sharp\xterm\leq1$ states that
$\xterm$ has at most one predecessor. We use $\sharp\xterm=0$ and
$\sharp\xterm=1$ as shorthands for $\neg(\sharp\xterm\geq1)$ and
$\sharp\xterm\geq1 \wedge \sharp\xterm\leq1$, respectively. The
formula below states that the heap can be decomposed into a list
segment starting with $\xterm$ and ending in $\yterm$, and several
disjoint cyclic lists:
\[\begin{array}{rcl}
\xterm \arrow{\circlearrowleft}{}^+ \yterm & \equiv & \sharp\xterm=0 
\wedge \mathsf{alloc}(\xterm) \wedge \sharp\yterm=1 \wedge
\neg\mathsf{alloc}(\yterm) ~\wedge \\ 
&& \forall \zterm ~.~ \zterm \not\teq \yterm \Rightarrow 
(\sharp\zterm=1 \Rightarrow \mathsf{alloc}(\zterm)) 
\wedge \forall \zterm ~.~ \sharp\zterm \leq 1
\end{array}\]
We forbid the existence of circular lists by adding the following
arithmetic constraint:
\[\forall \uterm \forall \tterm ~.~ \uterm \mapsto \tterm * \top \Rightarrow \uterm < \tterm
\tag{$\mathsf{nocyc}$}\label{eq:no-cycles}\] We ask, moreover, that the elements
of the list segment starting in $\xterm$ are consecutive perfect
squares:
\[\mathsf{consqr}(\xterm) \equiv \xterm=0 \wedge \xterm \mapsto 1 * \top \wedge 
\forall \zterm \forall \uterm \forall \tterm ~.~ \zterm \mapsto \uterm
* \uterm \mapsto \tterm * \top \Rightarrow
\tterm-\uterm=\uterm-\zterm+2
\tag{$\mathsf{consqr}$}\label{eq:consecutive-squares}\] Observe that
the formula $\exists\xterm\exists\yterm ~.~ \xterm
\arrow{\circlearrowleft}{}^+ \yterm \wedge \mathsf{nocyc} \wedge
\mathsf{consqr}(\xterm)$ belongs to
$\exists^*\forall^*\seplog(\lia)_{\Int,\Int}$.

\begin{theorem}\label{thm:il-undec}
  The satisfiability problem for
  $\exists^*\forall^*\seplog(\lia)_{\Int,\Int}$ is undecidable.
\end{theorem}
\proof{We use the same reduction as in the proof of Theorem
  \ref{thm:bsr-presburger}, with two differences: \begin{compactitem}
  \item we replace \ref{eq:sqr} by $\exists\xterm\exists\yterm ~.~
    \xterm \arrow{\circlearrowleft}{}^+ \yterm \wedge \mathsf{nocyc}
    \wedge \mathsf{consqr}(\xterm)$, and
  \item define $\theta_{x=y^2} \equiv \mathsf{alloc}(x) \wedge
    \mathsf{alloc}(x+2y+1) \wedge \forall z ~.~ x<z<x+2y+1 \Rightarrow
    \neg \mathsf{alloc}(z)$.\hfill\qed
  \end{compactitem}
}

It is tempting, at this point to ask whether interpreting locations as
integers and considering subsets of $\lia$ instead may help recover
the decidability. For instance, it has been found that the
Bernays-Sch\"onfinkel-Ramsey class is decidable in presence of
integers with difference bounds arithmetic \cite{WeidenbachVoigt15},
and the same type of question can be asked about the fragment of
$\exists^*\forall^*\seplog(\lia)_{\Int,\Int}$, with difference bounds
constraints only. 

Finally, we consider a variant of the previous undecidability result,
in which locations are the (uninterpreted) sort $U$ of $\euf$ and the
data consists of tuples of sort $U \times \Int$. This fragment of
$\seplog$ can be used to reason about lists with integer data. The
undecidability of this fragment can be proved along the same lines as
Theorem \ref{thm:il-undec}.

\begin{theorem}\label{thm:ulid-undec}
  The satisfiability problem for
  $\exists^*\forall^*\seplog(\euflia)_{U,U\times\Int}$ is undecidable.
\end{theorem}
\proof{
  Along the same lines as the proof of Theorem \ref{thm:il-undec}, with the following shorthands: 
  \[\begin{array}{rcl}
  \sharp\xterm\geq1 & \equiv & \exists\uterm^U \exists d^\Int ~.~ \uterm \mapsto (d,\xterm) * \top \\ 
  \sharp\xterm\leq1 & \equiv & \forall\uterm^U \forall\tterm^U \forall d^\Int ~.~ 
  \neg(\uterm \mapsto (d,\xterm) * \tterm \mapsto (d,\xterm) * \top) \\
  \\
  \mathsf{nocyc} & \equiv & \forall \uterm^U \forall \tterm^U \forall \vterm^U \forall d^\Int \forall e^\Int ~.~ 
  \uterm \mapsto (d,\tterm) * \tterm \mapsto (e,\vterm) * \top \Rightarrow d < e \\
  \\
  \mathsf{consqr}(\xterm) & \equiv & \exists \yterm^U \exists \zterm^U ~.~ \xterm\mapsto(0,\yterm) * \yterm\mapsto(1,\zterm) * \top \wedge \\
  && \forall \uterm^U \forall \tterm^U \forall \vterm^U \forall \wterm^U \forall d^\Int \forall e^\Int \forall f^\Int ~.~ 
  \uterm\mapsto(d,\tterm)*\tterm\mapsto(e,\vterm)*\vterm\mapsto(f,\wterm)*\top \Rightarrow \\
  && f-e=e-d+2 \hfill\text{\qed}
  \end{array}\]
}

\section{A Procedure for $\exists^*\forall^*$ Separation Logic in an SMT Solver}

This section presents a procedure for 
the satisfiability of $\exists^*\forall^*\seplog(\euf)_{U,U^k}$ inputs~\footnote{The procedure is incorporated into the master branch of the SMT solver CVC4 (\url{https://github.com/CVC4}), 
and can be enabled by command line parameter \texttt{-}\texttt{-quant-epr}.}.
Our procedure builds upon our previous work~\cite{ReynoldsIosifKingSerban16},
which gave a decision procedure for quantifier-free $\seplog(T)_{\locs,\data}$ inputs for theories $T$ where
the satisfiability problem for quantifier-free $T$-constraints is decidable.
Like existing approaches for quantified formulas in SMT~\cite{GeDeM-CAV-09,ReynoldsDKBT15Cav}, our approach is based on 
incremental quantifier instantiation based on a stream of candidate models returned by a solver for quantifier-free inputs.
Our approach for this fragment exploits the small model property given
in Lemma~\ref{lemma:small-model-property} to restrict the set of quantifier instantiations it considers
to a finite set.

\begin{figure}[t]
\begin{framed}
$\funcsolve( \exists \vec x\, \forall \vec y\, \varphi( \vec x, \vec y ) )$ where $\vec x = ( x_1, \ldots, x_m )$ and $\vec y = ( y_1, \ldots, y_n )$:
\begin{enumerate}
\item[\ ] Let $\vec k = ( k_1, \ldots, k_m )$ and $\vec e = ( e_1, \ldots, e_n )$ be fresh constants of the same type as $\vec x$ and $\vec y$.
\item[\ ] Let $L = L' \cup \{ k_1, \ldots, k_m \}$ where $L'$ is a set of fresh constants s.t. $\card{L'} = \len{\varphi( \vec x, \vec y )} + n$.
\item[\ ] Return $\funcsmtsolve( \exists \vec x\, \forall \vec y\, \varphi( \vec x, \vec y ), \emptyset, L )$.
\end{enumerate}
$\funcsmtsolve( \exists \vec x\, \forall \vec y\, \varphi( \vec x, \vec y ), \Gamma, L )$:
\begin{enumerate}
\item If $\Gamma$ is $(\seplog,\euf)$-unsat, return ``unsat".
\item Assume $\exists \vec x\, \forall \vec y\, \varphi( \vec x, \vec y )$ is equivalent to 
$\exists \vec x\, \forall \vec y\, \varphi_1( \vec x, \vec y ) \wedge \ldots \wedge \forall \vec y\, \varphi_p( \vec x, \vec y )$.
\item[\ ] If $\Gamma'_j = \Gamma \cup \{ \neg \varphi_j( \vec k, \vec e ) \wedge \displaystyle\bigwedge^n_{i=1} \bigvee_{t \in L} e_i \teq t \}$ is $(\seplog,\euf)$-unsat for all $j=1,\ldots,p$, return ``sat".
\item Otherwise, let $\I,h \models_{\tinyseplog} \Gamma'_j$ for some $j \in \{ 1,\ldots,p\}$.
\item[\ ] Let $\vec t = ( t_1, \ldots, t_n )$ be such that $e_i^\I = t_i^\I$ and $t_i \in L$ for each $i=1,\ldots,n$.
\item[\ ] Return $\funcsmtsolve( \exists \vec x\, \forall \vec y\, \varphi( \vec x, \vec y ), \Gamma \cup \{ \varphi_j( \vec k, \vec t ) \}, L )$.
\end{enumerate}
\end{framed}
\vspace{-2ex}
\caption{A counterexample-guided procedure for $\exists^*\forall^*\seplog(\euf)_{U,U^k}$ formulas $\exists \vec x\, \forall \vec y\, \varphi( \vec x, \vec y )$,
where $U$ is an uninterpreted sort in the signature of $\euf$.
\label{fig:smt-inst-sept}}
\end{figure}

Figure~\ref{fig:smt-inst-sept} gives a counterexample-guided approach for establishing the satisfiability of input $\exists \vec x\, \forall \vec y\, \varphi( \vec x, \vec y )$.
We first introduce tuples of fresh constants $\vec k$ and $\vec e$ of the same type as $\vec x$ and $\vec y$ respectively.
Our procedure will be based on finding a set of instantiations of $\forall \vec y\, \varphi( \vec k, \vec y )$ that are either collectively
unsatisfiable or are satisfiable and entail our input.
Then, we construct a set $L$ which is the union of constants $\vec k$ and a set $L'$ of fresh constants whose cardinality is equal to $\len{\varphi( \vec x, \vec y )}$ (see Section 3.1)
plus the number of universal variables $n$ in our input.
Conceptually, $L$ is a finite set of terms from which the instantiations of $\vec y$ in $\forall \vec y\, \varphi( \vec k, \vec y )$ can be built.

After constructing $L$, we call the recursive subprocedure $\funcsmtsolve$ on $\Gamma$ (initially empty) and $L$.
This procedure incrementally adds instances of $\forall \vec y\, \varphi( \vec k, \vec y )$ to $\Gamma$.
In step 1, we first check if $\Gamma$ is $(\seplog,T)$-unsatisfiable using the procedure from~\cite{ReynoldsIosifKingSerban16}.
If so, our input is $(\seplog,T)$-unsatisfiable.
Otherwise, in step 2 we consider the \emph{miniscoped} form of our input
$\exists \vec x\, \forall \vec y\, \varphi_1( \vec x, \vec y ) \wedge \ldots \wedge \forall \vec y\, \varphi_p( \vec x, \vec y )$,
that is, where quantification over $\vec x$ is distributed over conjunctions.
In the following, we may omit quantification on conjunctions $\varphi_j$ that do not contain variables from $\vec y$.
Given this formula, for each $j=1,\ldots,p$,
we check the $(\seplog,T)$-satisfiability of set $\Gamma'_j$ containing $\Gamma$, 
the negation of $\forall \vec y\, \varphi_j( \vec k, \vec y )$ where $\vec y$ is replaced by fresh contants $\vec e$, 
and a conjunction of constraints that says
each $e_i$ must be equal to at least one term in $L$ for $i=1,\ldots,n$.
If $\Gamma'_j$ is $(\seplog,T)$-unsatisfiable for each $j=1,\ldots,p$, our input is $(\seplog,T)$-satisfiable.
Otherwise in step 3, given an interpretation $\I$ and heap $h$ satisfying $\Gamma'_j$, we construct a tuple of terms $\vec t = ( t_1, \ldots, t_n )$
used for instantiating $\forall \vec y\, \varphi_j( \vec k, \vec y )$.
For each $i=1,\ldots,n$, we choose $t_i$ to be a term from $L$ whose interpretation is the same as $e_i$.
The existence of such a $t_i$ is guaranteed by the fact that $\I$ satisfies the constraint from $\Gamma'_j$ that tells us $e_i$ is equal to at least one such term.
This selection ensures that instantiations on each iteration
are chosen from a finite set of possibilities and are unique.
In practice, the procedure terminates, both for unsatisfiable and satisfiable inputs,
before considering all $\vec t$ from $L^n$ for each $\forall \vec y\, \varphi_j( \vec x, \vec y )$.

\begin{theorem}
\label{thm:solveslt}
Let $U$ be an uninterpreted sort belonging to the signature of $\euf$.
For all $\exists^*\forall^*\seplog(\euf)_{U,U^k}$ formulae $\psi$ of the form $\exists \vec x\, \forall \vec y\, \varphi( \vec x, \vec y )$,
$\funcsolve( \psi )$:
\begin{enumerate}
\item \label{it:sltunsat} Answers ``unsat" only if $\psi$ is $(\seplog,\euf)$-unsatisfiable.
\item \label{it:sltsat} Answers ``sat" only if $\psi$ is $(\seplog,\euf)$-satisfiable.
\item \label{it:sltterm} Terminates.
\end{enumerate}
\end{theorem}
\begin{proof}
To show (\ref{it:sltunsat}), note that $\Gamma$ contains only formulas of the form $\varphi_j( \vec k, \vec t )$,
which are consequences of our input.
Thus, when $\Gamma$ is $(\seplog,\euf)$-unsatisfiable, our input is $(\seplog,\euf)$-unsatisfiable as well.

To show (\ref{it:sltsat}), 
we have that $\Gamma$ is $(\seplog,\euf)$-satisfiable and
$\Gamma' = \Gamma \cup \{ \neg \varphi_j( \vec k, \vec e ) \wedge A \}$ is $(\seplog,\euf)$-unsatisfiable for each $j=1,\ldots,p$, where:
\[\begin{array}{c}
A = \displaystyle\bigwedge^n_{i=1} \bigvee_{t \in L} e_i \teq t
\end{array}\]
where $L = L' \cup \{ k_1, \ldots, k_m \}$ and $L'$ is a set of fresh constants s.t. $\card{L'} = \len{\varphi( \vec x, \vec y )} + n$.
In other words, we have that all models of $\Gamma$ satisfy $( \varphi_1( \vec k, \vec e ) \wedge \ldots \wedge \varphi_p( \vec k, \vec e ) ) \vee \neg A$,
which is equivalent to $\varphi( \vec k, \vec e )\vee \neg A$.
Since $\vec e$ is not contained in $\Gamma$, we have that all models of $\Gamma$ satisfy $\forall \vec y\, (\varphi( \vec k, \vec y ) \vee \neg A \{ \vec e \mapsto \vec y \})$.
Since $\Gamma$ is $(\seplog,\euf)$-satisfiable, we have that 
$\forall \vec y\, (\varphi( \vec k, \vec y ) \vee \neg A \{ \vec e \mapsto \vec y \})$ is $(\seplog,\euf)$-satisfiable as well.
Consider the formula $\forall \vec y\, \varphi( \vec k, \vec y )$.
By Lemma~\ref{lemma:small-model-property}, $\forall \vec y\, \varphi( \vec k, \vec y )$ has a model if and only if there exists an interpretation $\I$ and heap $h$
such that
$\I,h \models_{\tinyseplog} \forall \vec y\, \varphi( \vec k, \vec y )$ and
$\card{U^\I} \leq \len{\varphi} + \card{\fv{\forall \vec y\,\varphi( \vec x, \vec y )}} + n$.
Due to the construction of $L$, this implies that $\forall \vec y\, (\varphi( \vec k, \vec y ) \vee \neg A \{ \vec e \mapsto \vec y \})$ is $(\seplog,\euf)$-satisfiable
if and only if $\forall \vec y\, \varphi( \vec k, \vec y )$ is $(\seplog,\euf)$-satisfiable.
Thus, $\exists \vec x\, \forall \vec y\, \varphi( \vec x, \vec y )$ is $(\seplog,\euf)$-satisfiable.

To show (\ref{it:sltterm}), 
clearly only a finite number of possible formulas can be added to $\Gamma$ as a result of the procedure,
since all terms $\vec t$ belong to the finite set $L$ and $p$ is finite.
Furthermore, on every iteration, for any $j$, $\I$ satisfies $\Gamma$ and $\neg \varphi_j( \vec k, \vec e )$.
Since $e_i^\I = t_i^\I$ for each $i = 1, \ldots, n$, we have that $\varphi_j( \vec k, \vec t ) \not\in \Gamma$,
and thus a new formula is added to $\Gamma$ on every call.
Thus, only a finite number of recursive calls are made to $\funcsmtsolve$.
Since the $(\seplog,\euf)$-satisfiability of quantifier-free is decidable, all steps in the procedure are terminating,
and thus $\funcsolve$ terminates.
\qed
\end{proof}

We discuss a few important details regarding our implementation of the procedure.

%\paragraph{Miniscoping}
%In practice, the universal quantification in our input $\exists \vec x\, \forall \vec y\, \varphi( \vec x, \vec y )$
%may be \emph{miniscoped} to obtain a formula of the form $\exists \vec x\, \psi_1 \wedge \ldots \wedge \psi_k$,
%where each $\psi_i$ is of the form $\forall \vec z\, \varphi_i( \vec x, \vec z )$ for a (possibly empty) tuple $\vec z$ of variables that are a subset of those in $\vec y$,
%and $\varphi_i$ is quantifier-free.
%Our procedure then instantiates each of these quantified formulas independently in the manner described in Figure~\ref{fig:smt-inst-sept},
%that is, on each iteration an instance of some $\forall \vec z\, \varphi_i( \vec x, \vec z )$ is added to $\Gamma$.
%Our procedure terminates with ``sat" only when $\Gamma \cup \{ \varphi_i( \vec k, \vec z ) \{ \vec y \mapsto \vec e \} \wedge \displaystyle\wedge^n_{i=1} \vee_{t \in L} e_i \teq t \}$ is 
%$(\seplog,T)$-unsatisfiable for all $i=1,\ldots,k$.

\paragraph{Matching Heuristics} 
When constructing the terms $\vec t$ for instantiation,
it may be the case that $e_i^\I = u^\I$ for multiple $u \in L$ for some $i \in \{ 1,\ldots,n \}$.
In such cases, the procedure will choose one such $u$ for instantiation.
To increase the likelihood of the instantiation being relevant to the satisfiability of our input,
we use heuristics for selecting the best possible $u$ among those whose interpretation is equal to $e_i$ in $\I$.
In particular, if $e_i^\I = u_1^\I = u_2^\I$, and $\Gamma'$ contains predicates of the form
$e_i \mapsto v$ and $u_1 \mapsto v_1$ for some $v, v_1$ where $v^\I = v_1^\I$
but no predicate of the form $u_2 \mapsto v_2$ for some $v_2$ where $v^\I = v_2^\I$,
then we strictly prefer term $u_1$ over term $u_2$ when choosing term $t_i$ for $e_i$.

\paragraph{Finding Minimal Models}
Previous work~\cite{ReyEtAl-1-RR-13} developed efficient techniques for finding small models for uninterpreted sorts in CVC4.
We have found these techniques to be beneficial to the performance of the procedure in Figure~\ref{fig:smt-inst-sept}.
In particular, we use these techniques to find $\Sigma$-interpretations $\I$ in $\funcsmtsolve$ that
interpret $U$ as a finite set of minimal size.
When combined with the aforementioned matching heuristics,
these techniques lead to finding useful instantiations more quickly, since more terms are constrained to be equal to $e_i$ for $i=1,\ldots,n$
in interpretations $\I$.

\paragraph{Symmetry Breaking} 
The procedure in Figure~\ref{fig:smt-inst-sept} introduces a set of fresh constants $L$,
which in turn introduce the possibility of discovering $\Sigma$-interpretations $\I$ that are isomorphic, that is,
identical up to renaming of constants in $L'$.
Our procedure adds additional constraints to $\Gamma$ that do not affect its satisfiability,
but reduce the number of isomorphic models.
In particular, we consider an ordering $\prec$ on the constants from $L'$,
and add constraints that ensure that all models $(\I,h)$ of $\Gamma$
are such that if $\ell_1^\I \not\in \dom(h)$, then $\ell_2^\I \not\in \dom(h)$ for all $\ell_2$ such that $\ell_1 \prec \ell_2$.

\begin{example}
Say we wish to show the validity of the entailment $\xterm \neq \yterm
\wedge \xterm \mapsto \zterm \models_{\tinyseplog} \exists \uterm ~.~ \xterm \mapsto \uterm$, 
from the introductory example (Section
\ref{sec:intro}), where $\xterm,\yterm,\zterm,\uterm$ are of
sort $U$ of $\euf$. This entailment is valid iff the
$\exists^*\forall^*\seplog(\euf)_{U,U^k}$ formula $\exists \xterm
\exists \yterm \exists \zterm \forall \uterm ~.~ \xterm \not\teq
\yterm \wedge \xterm \mapsto \zterm \wedge \neg \xterm \mapsto \uterm$
is $(\seplog,\euf)$-unsatisfiable.  A run of the procedure in
Figure~\ref{fig:smt-inst-sept} on this input constructs tuples $\vec k
= ( k_x, k_y, k_z )$ and $\vec e = ( e_u )$, and set $L = \{ k_x, k_y,
k_z,\ell_1,\ell_2 \}$, noting that $\len{\xterm \not\teq \yterm \wedge
  \xterm \mapsto \zterm \wedge \neg \xterm \mapsto \uterm}=1$.  We
then call $\funcsmtsolve$ where $\Gamma$ is initially empty.  By
miniscoping, our input is equivalent to $\exists \xterm \exists \yterm
\exists \zterm ~.~ \xterm \not\teq \yterm \wedge \xterm \mapsto \zterm
\wedge \forall \uterm ~.~ \neg \xterm \mapsto \uterm$.
%and thus we consider instantiations of $k_x \not\teq k_y$, $k_x \mapsto k_z$, and $\forall u\, \neg k_x \mapsto u$ independently.
On the first two recursive calls to $\funcsmtsolve$, we may add $k_x
\not\teq k_y$ and $k_x \mapsto k_z$ to $\Gamma$ by trivial
instantiation of the first two conjuncts.  On the third recursive
call, $\Gamma$ is $(\seplog,\euf)$-satisfiable, and we check the
satisfiability of:
\[\begin{array}{c}
\Gamma' = \{ k_x \neq k_y, k_x \mapsto k_z, k_x \mapsto e_u \wedge ( e_u \teq k_x \vee e_u \teq k_y \vee e_u \teq k_z\vee e_u \teq \ell_1\vee e_u \teq \ell_2 ) \}
\end{array}\]
Since $k_x \mapsto k_z$ and $k_x \mapsto e_u$ are in $\Gamma'$, all 
$\Sigma$-interpretations $\I$ and heaps $h$ such that $\I,h \models_{\tinyseplog} \Gamma'$ are such that $e_u^\I =k_z^\I$.
Since $k_z \in L$, we may choose to add the instantiation 
$\neg k_x \mapsto k_z$ to $\Gamma$, after which $\Gamma$ is $(\seplog,\euf)$-unsatisfiable
on the next recursive call to $\funcsmtsolve$.
Thus, our input is $(\seplog,\euf)$-unsatisfiable and the entailment is valid. \hfill$\blacksquare$
\end{example}

A modified version of the procedure in Figure~\ref{fig:smt-inst-sept} can be used for
$\exists^*\forall^*\seplog(T)_{\locs,\data}$-satisfiability for theories $T$ beyond equality,
and where $\locs$ and $\data$ are not restricted to uninterpreted sorts.
Notice that in such cases, we cannot restrict $\Sigma$-interpretations $\I$ in $\funcsmtsolve$ to
interpret each $e_i$ as a member of finite set $L$, 
and hence we modify $\funcsmtsolve$ to omit the constraint restricting variables in $\vec e$ to be equal to a term from $L$ in the check in Step 2.
This modification results in a procedure that is sound both for ``unsat" and ``sat",
but is no longer terminating in general.
Nevertheless, it may be used as a heuristic for determining $\exists^*\forall^*\seplog(T)_{\locs,\data}$-(un)satisfiability.

\section{Experimental Evaluation}
We implemented the $\funcsolve$ procedure from
Figure~\ref{fig:smt-inst-sept} within the CVC4 SMT
solver\footnote{Available at \url{http://cvc4.cs.nyu.edu/web/}.}
(version 1.5 prerelease). This implementation was tested on two kinds
of benchmarks:\begin{inparaenum}[(i)] \item finite unfoldings of
  inductive predicates, mostly inspired by benchmarks used in the
  SL-COMP'14 solver competition \cite{sl-comp14}, and \item
  verification conditions automatically generated by applying the
  weakest precondition calculus of ~\cite{IshtiaqOHearn01} to the
  program loops in Figure \ref{fig:loops}. \end{inparaenum} All
experiments were run on a 2.80GHz Intel(R) Core(TM) i7 CPU machine
with 8MB of cache \footnote{The CVC4 binary and examples used in
  these experiments are available at \\
  {\url{http://cs.uiowa.edu/~ajreynol/VMCAI2017-seplog-epr}}.}.

\begin{figure}[htb]
\begin{center}
\begin{minipage}{7cm}
{\footnotesize\begin{tabbing}
1:  \texttt{wh}\=\texttt{ile} $\wterm \neq \nil$ \texttt{do} \\
2:  \> $\mathbf{assert}(\wterm.\datap = c_0)$ \\
3:  \> $\vterm := \wterm$; \\
4:  \> $\wterm := \wterm.\nextp$; \\
5:  \> $\mathbf{dispose}(\vterm)$;   \\
6:  \> \texttt{do} \\
\\
\> {\bf (z)disp}
\end{tabbing}}
\end{minipage}
\hspace*{2cm}
\begin{minipage}{6.5cm}
{\footnotesize\begin{tabbing}
1:  \texttt{wh}\=\texttt{ile} $\uterm \neq \nil$ \texttt{do} \\
2:  \> $\mathbf{assert}(\uterm.\datap = c_0)$ \\
3:  \> $\wterm := \uterm.\nextp$; \\
4:  \> $\uterm.\nextp := \vterm$; \\
5:  \> $\vterm := \uterm$;   \\
6:  \> $\uterm := \wterm$;   \\
7:  \> \texttt{do} \\
\> {\bf (z)rev}
\end{tabbing}}
\end{minipage}
\begin{tabular}{r@{ }c@{ }l@{\hskip 1.5cm}r@{ }c@{ }l}
$\mathsf{list}^0(x)$ & $\triangleq$ & $\emp \land x = \nil$ & $\mathsf{zlist}^0(x)$ & $\triangleq$ & $\emp \land x = \nil$\\
$\mathsf{list}^n(x)$ & $\triangleq$ & $\exists y \, . \, x \mapsto y * \mathsf{list}^{n-1}(y)$ & $\mathsf{zlist}^n(x)$ & $\triangleq$ & $\exists y \, . \, x \mapsto (c_0, y) * \mathsf{zlist}^{n-1}(y)$
\end{tabular}
\end{center}
\vspace*{-\baselineskip}
\caption{Program Loops}\label{fig:loops}
\vspace*{-\baselineskip}
\end{figure}

We compared our implementation with the results of applying the CVC4
decision procedure for the quantifier-free fragment of $\seplog$
\cite{ReynoldsIosifKingSerban16} to a variant of the benchmarks,
obtained by manual quantifier instantiation, as follows. Consider
checking the validity of the entailment $\exists \vec{x} ~.~
\phi(\vec{x}) \models_{\tinyseplog} \exists \vec{y} ~.~
\psi(\vec{y})$, which is equivalent to the unsatisfiability of the
formula \(\exists \vec{x} \forall \vec{y} ~.~ \phi(\vec{x}) \wedge
\neg \psi(\vec{y})\). We first check the satisfiability of $\phi$. If
$\phi$ is not satisfiable, the entailment holds trivially, so let us
assume that $\phi$ has a model. Second, we check the satisfiability of
$\phi \wedge \psi$. Again, if this is unsatisfiable, the entailment
cannot hold, because there exists a model of $\phi$ which is not a
model of $\psi$. Else, if $\phi \wedge \psi$ has a model, we add an
equality $x=y$ for each pair of variables $(x,y) \in
\vec{x}\times\vec{y}$ that are mapped to the same term in this model,
the result being a conjunction $E(\vec{x},\vec{y})$ of
equalities. Finally, we check the satisfiability of the formula $\phi
\wedge \neg\psi \wedge E$. If this formula is unsatisfiable, the
entailment is valid, otherwise, the check is inconclusive. The times
in Table \ref{tab:experiments} correspond to checking satisfiability
of \(\exists \vec{x} \forall \vec{y} ~.~ \phi(\vec{x}) \wedge \neg
\psi(\vec{y})\) using the $\funcsolve$ procedure (Figure
\ref{fig:smt-inst-sept}), compared to checking satisfiability of $\phi
\wedge \neg\psi \wedge E$, where $E$ is manually generated.

%% \begin{example}
%% Applying this method to show the validity of the entailment $\xterm \neq \yterm
%% \wedge \xterm \mapsto \zterm \models_{\tinyseplog} \exists \uterm \, . \,
%% \xterm \mapsto \uterm$, we first check the satisfiability of $\xterm \neq \yterm
%% \wedge \xterm \mapsto \zterm$. It is indeed satisfiable, so we proceed by checking the satisfiability of $\xterm \neq \yterm
%% \wedge \xterm \mapsto \zterm \land \xterm \mapsto \uterm$. Again, the formula is satisfiable and yields a model where $\zterm$ and $\uterm$ are equal. Finally, we check the satisfiability of $\xterm \neq \yterm \wedge \xterm \mapsto \zterm \land \lnot \xterm \mapsto \uterm \land \zterm = \uterm$. Since this formula is unsatisfiable, we can conclude that the initial entailment is valid.
%% \end{example}

In the first set of experiments (Table \ref{tab:experiments}) we have
considered inductive predicates commonly used as verification
benchmarks \cite{sl-comp14}. Here we check the validity of the
entailment between $\mathsf{lhs}$ and $\mathsf{rhs}$, where both
predicates are unfolded $n=1,2,3,4,8$ times. The entailment between
$\mathsf{pos}_2^1$ and $\mathsf{neg}_4^1$ is skipped because it is not
valid (since the negated formula is satisfiable, we cannot generate
the manual instantiation). 

The second set of experiments considers the verification conditions of
the forms $\varphi \Rightarrow \mathsf{wp}(\mathbf{l},\phi)$ and
$\varphi \Rightarrow \mathsf{wp}^n(\mathbf{l},\phi)$, where
$\mathsf{wp}(\mathbf{l},\phi)$ denotes the weakest precondition of the
$\seplog$ formula $\phi$ with respect to the sequence of statements
$\mathbf{l}$, and $\mathsf{wp}^n(\mathbf{l},\phi) =
\mathsf{wp}(\mathbf{l},
\ldots\mathsf{wp}(\mathbf{l}, \\ \mathsf{wp}(\mathbf{l},\phi)) \ldots)$
denotes the iterative application of the weakest precondition $n$
times in a row. We consider the loops depicted in Figure
\ref{fig:loops}, where, for each loop {\bf l}, we consider the variant
    {\bf zl} as well, which tests that the data values contained
    within the memory cells are equal to a constant $c_0$ of sort
    $\locs$, by the assertions on line 2. The postconditions are
    specified by finite unfoldings of the inductive predicates
    $\mathsf{list}$ and $\mathsf{zlist}$.

We observed that, compared to checking the manual instantiation, 
the fully automated solver was less than $0.5$ seconds
slower on $72\%$ of the test cases, and less than $1$ second slower on $79\%$ 
of the test cases.
The automated solver experienced $3$ timeouts, where the manual
instantiation succeeds (for $\widehat{\mathsf{tree}}$ vs
$\mathsf{tree}$ with $n=8$, $\widehat{\mathsf{ts}}$ vs
$\mathsf{ts}$ with $n=3$, and $\mathsf{list}^{n}(u) * \mathsf{list}^{0}(v)$ 
vs $\mathsf{wp}^n(\mathbf{rev}, u = \nil \land \mathsf{list}^{n}(v))$
with $n=8$). These timeouts are caused by the first call to the
quantifier-free $\seplog$ decision procedure, which fails to produce a
model in less than $300$ seconds (time not accounted for in the
manually produced instance of the problem).

\begin{table}[htb]
\begin{tabular}{|c|c|c|c|c|c|c|c|}
\Xhline{2\arrayrulewidth}
{\bf lhs} & {\bf rhs} & & $n=1$ & $n=2$ & $n=3$ & $n=4$ & $n=8$\\
\Xhline{2\arrayrulewidth}
\multicolumn{8}{|c|}{\bf Unfoldings of inductive predicates}\\
% lseg ----------------------------------------------------------
\Xhline{2\arrayrulewidth}
$\scriptstyle \widehat{\mathsf{ls}}(x, y) \triangleq \emp \land x = y \lor$ & $\scriptstyle \mathsf{ls}(x, y) \triangleq \emp \land x = y \lor $ & {\scriptsize $\funcsolve$} & {\scriptsize $<0.01$s} & {\scriptsize 0.02s} & {\scriptsize 0.03s} & {\scriptsize 0.05s} & {\scriptsize 0.21s}\\[-1pt]
\cline{3-8}
$\scriptstyle  \exists z \, . \, x \neq y \land x \mapsto z * \widehat{\mathsf{ls}}(z, y)$ & $\scriptstyle \exists z \, . \, x \mapsto z * \mathsf{ls}(z, y)$ & {\scriptsize manual} & {\scriptsize $<0.01$s} & {\scriptsize $<0.01$s} & {\scriptsize $<0.01$s} & {\scriptsize $<0.01$s} & {\scriptsize $<0.01$s}\\[-1pt]
% tree ----------------------------------------------------------
\Xhline{2\arrayrulewidth}
$\scriptstyle \widehat{\mathsf{tree}}(x) \triangleq \emp \land x = \mathsf{nil} \lor$ & $\scriptstyle \mathsf{tree}(x) \triangleq \emp \land x = \mathsf{nil} \lor$ & {\scriptsize $\funcsolve$} & {\scriptsize $<0.01$s} & {\scriptsize 0.04s} & {\scriptsize 1.43s} & {\scriptsize 23.42s} & {\scriptsize $>300$s}\\[-1pt]
\cline{3-8}
$\scriptstyle \exists l \exists r \, . \, l \neq r \land x \mapsto (l, r) * \mathsf{tree}(l) * \mathsf{tree}(r)$ & $\scriptstyle \exists l \exists r \, . \, x \mapsto (l, r) * \mathsf{tree}(l) * \mathsf{tree}(r)$ & {\scriptsize manual} & {\scriptsize $<0.01$s} & {\scriptsize $<0.01$s} & {\scriptsize $<0.01$s} & {\scriptsize $<0.01$s} & {\scriptsize 0.09s}\\[-1pt]
% tseg ----------------------------------------------------------
\Xhline{2\arrayrulewidth}
$\scriptstyle \widehat{\mathsf{ts}}(x, a) \triangleq \emp \land x = \nil \lor $ & $\scriptstyle\mathsf{ts}(x, a) \triangleq \emp \land x = \nil \lor$ & {\scriptsize $\funcsolve$} & {\scriptsize $<0.01$s} & {\scriptsize 0.81s} & {\scriptsize $>300$s} & {\scriptsize $>300$s} & {\scriptsize $>300$s}\\[-1pt]
\cline{3-8}
$\scriptstyle \exists l \exists r \, . \, x \neq y \land x \mapsto (l, r) * \widehat{\mathsf{ts}} (l, y) * \mathsf{tree}(r) \lor$ & $\scriptstyle \exists l \exists r \, . \, \land x \mapsto (l, r) * \mathsf{ts} (l, y) * \mathsf{tree}(r) \lor$ & {\scriptsize manual} & {\scriptsize $<0.01$s} & {\scriptsize 0.03s} & {\scriptsize 103.89s} & {\scriptsize $>300$s} & {\scriptsize $>300$s}\\[-1pt]
$\scriptstyle \exists l \exists r \, . \, x \neq y \land x \mapsto (l, r) * \mathsf{tree} (l) * \widehat{\mathsf{ts}}(r, y)$ & $\scriptstyle \exists l \exists r \, . \, \land x \mapsto (l, r) * \mathsf{tree} (l) * \mathsf{ts}(r, y)$ & & & & & &\\
% unfold-sat ----------------------------------------------------
\Xhline{2\arrayrulewidth}
$\scriptstyle \mathsf{pos_1}(x, a) \triangleq x \mapsto a \lor \exists y \exists b \, . \,$ & $\scriptstyle \mathsf{neg_1}(x, a) \triangleq \lnot x \mapsto a \lor \exists y \exists b \, . \,$ & {\scriptsize $\funcsolve$} & {\scriptsize 0.34s} & {\scriptsize 0.01s} & {\scriptsize 0.31s} & {\scriptsize 0.76s} & {\scriptsize 21.19s}\\[-1pt]
\cline{3-8}
$\scriptstyle x \mapsto a * \mathsf{pos_1}(y, b)$ & $\scriptstyle x \mapsto a * \mathsf{neg_1}(y, b)$ & {\scriptsize manual} & {\scriptsize 0.04s} & {\scriptsize 0.05s} & {\scriptsize 0.08s} & {\scriptsize 0.12s} & {\scriptsize 0.53s}\\[-1pt]
% unfold-unsat --------------------------------------------------
\Xhline{2\arrayrulewidth}
$\scriptstyle \mathsf{pos_1}(x, a) \triangleq x \mapsto a \lor \exists y \exists b \, . \,$ & $\scriptstyle \mathsf{neg_2}(x, a) \triangleq x \mapsto a \lor \exists y \exists b \, . \,$ & {\scriptsize $\funcsolve$} & {\scriptsize 0.03s} & {\scriptsize 0.12s} & {\scriptsize 0.23s} & {\scriptsize 0.46s} & {\scriptsize 3.60s}\\[-1pt]
\cline{3-8}
$\scriptstyle x \mapsto a * \mathsf{pos_1}(y, b)$ & $\scriptstyle \lnot x \mapsto a * \mathsf{neg_2}(y, b)$ & {\scriptsize manual} & {\scriptsize 0.05s} & {\scriptsize 0.08s} & {\scriptsize 0.08s} & {\scriptsize 0.12s} & {\scriptsize 0.54s}\\[-1pt]
% chain-sat -----------------------------------------------------
\Xhline{2\arrayrulewidth}
$\scriptstyle \mathsf{pos_2}(x, a) \triangleq x \mapsto a \lor \exists y \, . \,$ & $\scriptstyle \mathsf{neg_3}(x, a) \triangleq \lnot x \mapsto a \lor \exists y \, . \,$ & {\scriptsize $\funcsolve$} & {\scriptsize 0.04s} & {\scriptsize 0.13s} & {\scriptsize 0.28s} & {\scriptsize 0.48s} & {\scriptsize 4.20s}\\[-1pt]
\cline{3-8}
$\scriptstyle x \mapsto a * \mathsf{pos_2}(a, y)$ & $\scriptstyle x \mapsto a * \mathsf{neg_3}(a, y)$ & {\scriptsize manual} & {\scriptsize 0.01s} & {\scriptsize 0.03s} & {\scriptsize 0.05s} & {\scriptsize 0.09s} & {\scriptsize 0.45s}\\[-1pt]
% chain-unsat ---------------------------------------------------
\Xhline{2\arrayrulewidth}
$\scriptstyle \mathsf{pos_2}(x, a) \triangleq x \mapsto a \lor \exists y \, . \,$ & $\scriptstyle \mathsf{neg_4}(x, a) \triangleq x \mapsto a \lor \exists y \, . \,$ & {\scriptsize $\funcsolve$} & --- & {\scriptsize 0.08s} & {\scriptsize 0.15s} & {\scriptsize 0.26s} & {\scriptsize 1.33s}\\[-1pt]
\cline{3-8}
$\scriptstyle x \mapsto a * \mathsf{pos_2}(a, y)$ & $\scriptstyle \lnot x \mapsto a * \mathsf{neg_4}(a, y)$ & {\scriptsize manual} & --- & {\scriptsize 0.03s} & {\scriptsize 0.06s} & {\scriptsize 0.09s} & {\scriptsize 0.46s}\\[-1pt]
\Xhline{2\arrayrulewidth}
\multicolumn{8}{|c|}{\bf Verification conditions}\\
% dispose -------------------------------------------------------
\Xhline{2\arrayrulewidth}
$\scriptstyle \mathsf{list}^{n}(w)$ & $\scriptstyle\mathsf{wp}(\mathbf{disp}, \mathsf{list}^{n-1}(w))$ & {\scriptsize $\funcsolve$} & {\scriptsize 0.01s} & {\scriptsize 0.03s} & {\scriptsize 0.08s} & {\scriptsize 0.19s} & {\scriptsize 1.47s}\\[-1pt]
\cline{3-8}
& & {\scriptsize manual} & {\scriptsize $<0.01$s} & {\scriptsize 0.01s} & {\scriptsize 0.02s} & {\scriptsize 0.05s} & {\scriptsize 0.26s}\\[-1pt]
% dispose-iter --------------------------------------------------
\Xhline{2\arrayrulewidth}
$\scriptstyle\mathsf{list}^{n}(w)$ & $\scriptstyle\mathsf{wp}^{n}(\mathbf{disp},\emp \land w = \nil)$ & {\scriptsize $\funcsolve$} & {\scriptsize 0.01s} & {\scriptsize 0.06s} & {\scriptsize 0.17s} & {\scriptsize 0.53s} & {\scriptsize 7.08s}\\[-1pt]
\cline{3-8}
& & {\scriptsize manual} & {\scriptsize $<0.01$s} & {\scriptsize 0.02s} & {\scriptsize 0.08s} & {\scriptsize 0.14s} & {\scriptsize 2.26s}\\[-1pt]
% test-dispose --------------------------------------------------
\Xhline{2\arrayrulewidth}
$\scriptstyle\mathsf{zlist}^{n}(w)$  & $\scriptstyle\mathsf{wp}(\textbf{zdisp}, \mathsf{zlist}^{n-1}(w))$ & {\scriptsize $\funcsolve$} & {\scriptsize 0.04s} & {\scriptsize 0.05s} & {\scriptsize 0.09s} & {\scriptsize 0.19s} & {\scriptsize 1.25s}\\[-1pt]
\cline{3-8}
& & {\scriptsize manual} & {\scriptsize $<0.01$s} & {\scriptsize 0.01s} & {\scriptsize 0.02s} & {\scriptsize 0.04s} & {\scriptsize 0.29s}\\[-1pt]
% test-dispose-iter ---------------------------------------------
\Xhline{2\arrayrulewidth}
$\scriptstyle\mathsf{zlist}^{n}(w)$ & $\scriptstyle\mathsf{wp}^n(\textbf{zdisp}, \emp \land w = \nil)$ & {\scriptsize $\funcsolve$} & {\scriptsize 0.01s} & {\scriptsize 0.10s} & {\scriptsize 0.32s} & {\scriptsize 0.87s} & {\scriptsize 11.88s}\\[-1pt]
\cline{3-8}
& & {\scriptsize manual} & {\scriptsize 0.01s} & {\scriptsize 0.02s} & {\scriptsize 0.07s} & {\scriptsize 0.15s} & {\scriptsize 2.20s}\\[-1pt]
% rev -----------------------------------------------------------
\Xhline{2\arrayrulewidth}
$\scriptstyle\mathsf{list}^{n}(u) * \mathsf{list}^{0}(v)$ & $\scriptstyle\mathsf{wp}(\mathbf{rev}, \mathsf{list}^{n-1}(u) * \mathsf{list}^{1}(v))$ & {\scriptsize $\funcsolve$} & {\scriptsize 0.38s} & {\scriptsize 0.06s} & {\scriptsize 0.11s} & {\scriptsize 0.16s} & {\scriptsize 0.56s}\\[-1pt]
\cline{3-8}
& & {\scriptsize manual} & {\scriptsize 0.07s} & {\scriptsize 0.03s} & {\scriptsize 0.07s} & {\scriptsize 0.11s} & {\scriptsize 0.43s}\\[-1pt]
% rev-iter ------------------------------------------------------
\Xhline{2\arrayrulewidth}
$\scriptstyle\mathsf{list}^{n}(u) * \mathsf{list}^{0}(v)$ & $\scriptstyle\mathsf{wp}^n(\mathbf{rev}, u = \nil \land \mathsf{list}^{n}(v))$ & {\scriptsize $\funcsolve$} & {\scriptsize 0.38s} & {\scriptsize 0.07s} & {\scriptsize 0.30s} & {\scriptsize 68.68s} & {\scriptsize $>300$s}\\[-1pt]
\cline{3-8}
& & {\scriptsize manual} & {\scriptsize 0.08s} & {\scriptsize 0.06s} & {\scriptsize 0.11s} & {\scriptsize 0.23s} & {\scriptsize 1.79s}\\[-1pt]
% test-rev ------------------------------------------------------
\Xhline{2\arrayrulewidth}
$\scriptstyle\mathsf{zlist}^{n}(u) * \mathsf{zlist}^{0}(v)$ & $\scriptstyle\mathsf{wp}(\mathbf{zrev}, \mathsf{zlist}^{n-1}(u) * \mathsf{zlist}^{1}(v))$ & {\scriptsize $\funcsolve$} & {\scriptsize 0.22s} & {\scriptsize 0.07s} & {\scriptsize 0.15s} & {\scriptsize 0.21s} & {\scriptsize 0.75s}\\[-1pt]
\cline{3-8}
& & {\scriptsize manual} & {\scriptsize 0.04s} & {\scriptsize 0.02s} & {\scriptsize 0.04s} & {\scriptsize 0.06s} & {\scriptsize 0.31s}\\[-1pt]
% test-rev-iter -------------------------------------------------
\Xhline{2\arrayrulewidth}
$\scriptstyle\mathsf{zlist}^{n}(u) * \mathsf{zlist}^{0}(v)$ & $\scriptstyle\mathsf{wp}^n(\mathbf{zrev}, u = \nil \land \mathsf{zlist}^{n}(v))$ & {\scriptsize $\funcsolve$} & {\scriptsize 0.23s} & {\scriptsize 0.09s} & {\scriptsize 0.17s} & {\scriptsize 0.30s} & {\scriptsize 2.06s}\\[-1pt]
\cline{3-8}
& & {\scriptsize manual} & {\scriptsize 0.04s} & {\scriptsize 0.02s} & {\scriptsize 0.05s} & {\scriptsize 0.09s} & {\scriptsize 0.48s}\\[-1pt]
\Xhline{2\arrayrulewidth}
\end{tabular}\\
\caption{Experimental results}\label{tab:experiments}
\vspace*{-1cm}
\end{table}

\section{Conclusions and Future Work}

We present theoretical and practical results for the existence of
effective decision procedures for the fragment of Separation Logic
obtained by restriction of formulae to quantifier prefixes in the set
$\exists^*\forall^*$. The theoretical results range from
undecidability, when the set of memory locations is taken to be the
set of integers and linear arithmetic constraints are allowed, to
\textsc{PSPACE}-completeness, when locations and data in the cells
belong to an uninterpreted sort, equipped with equality only. We have
implemented a decision procedure for the latter case in the CVC4 SMT
solver, using an effective counterexample-driven instantiation of the
universal quantifiers. The procedure is shown to be sound, complete
and termination is guaranteed when the input belongs to a decidable
fragment of $\seplog$.

As future work, we aim at refining the decidability chart for
$\exists^*\forall^*\seplog(T)_{\locs,\data}$, by considering the case
where the locations are interpreted as integers, with weaker
arithmetics, such as sets of difference bounds, or octagonal
constraints. These results are likely to extend the application range
of our tool, to e.g.\ solvers working on $\seplog$ with inductive
definitions and data constraints. The current implementation should
also benefit from improvements of the underlying quantifier-free
$\seplog$ and set theory solvers.

%%%%%%%%%%%%%%%%%%%%%%%%%%%%%%%%%%%%%%%%%%%%%%%%%%%%%%%%%%%%%%%%%%%%%%%%%%%%%%%
\bibliographystyle{splncs03} 
\bibliography{refs}
%%%%%%%%%%%%%%%%%%%%%%%%%%%%%%%%%%%%%%%%%%%%%%%%%%%%%%%%%%%%%%%%%%%%%%%%%%%%%%%

%%%%%%%%%%%%%%%%%%%%%%%%%%%%%%%%%%%%%%%%%%%%%%%%%%%%%%%%%%%%%%%%%%%%%%%%%%%%%%%
\end{document}